\documentclass{article}

\usepackage{arxiv}

\usepackage[utf8]{inputenc} % allow utf-8 input
\usepackage[T1]{fontenc}    % use 8-bit T1 fonts
\usepackage{url}            % simple URL typesetting
\usepackage{booktabs}       % professional-quality tables
\usepackage{amsfonts}       % blackboard math symbols
\usepackage{nicefrac}       % compact symbols for 1/2, etc.
\usepackage{microtype}      % microtypography
\usepackage{color}
\usepackage{lipsum}
\usepackage{amsthm,amssymb}
\usepackage{indentfirst}
\usepackage{graphicx}
\usepackage{epstopdf}
\usepackage{fourier}
\usepackage{bm}
\usepackage{mathtools}
\usepackage{latexsym,enumerate}
\usepackage{multicol}
\usepackage[skip=2pt]{caption}
\usepackage[font=small,skip=0pt]{subcaption}
\newcommand{\comment}[1]{}

\newcommand{\BEA}{\begin{eqnarray}}
\newcommand{\EEA}{\end{eqnarray}}

\newtheorem{thm}{Theorem}[section]
\newtheorem{prop}[thm]{Proposition}

\title{Kernel-based Prediction of Non-Markovian Time Series}

\author{
  Faheem Gilani \\
  Department of Mathematics \\
  The Pennsylvania State University, University Park, PA 16802, USA\\
  \texttt{fhg3@psu.edu} \\  
  \And
  Dimitrios Giannakis \\
  Courant Institute of Mathematical Sciences \\
  New York University, New York, NY 10012, USA\\
  \texttt{dimitris@cims.nyu.edu}   
  \And
  John Harlim \\
  Department of Mathematics, Department of Meteorology and Atmospheric Science, \\ Institute for Computational and Data Sciences \\
  The Pennsylvania State University, University Park, PA 16802, USA\\
  \texttt{jharlim@psu.edu} 
}

\begin{document}

\maketitle

\begin{abstract}
A nonparametric method to predict non-Markovian time series of partially observed dynamics is developed. The prediction problem we consider is a supervised learning task of finding a regression function that takes a delay-embedded observable to the observable at a future time. When delay-embedding theory is applicable, the proposed regression function is a consistent estimator of the flow map induced by the delay-embedding. Furthermore, the corresponding Mori-Zwanzig equation governing the evolution of the observable simplifies to only a Markovian term, represented by the regression function. We realize this supervised learning task with a class of kernel-based linear estimators, the kernel analog forecast (KAF),
which are consistent in the limit of large data. In a scenario with a high-dimensional covariate space, we employ a Markovian kernel smoothing method which is computationally cheaper than the Nystr\"{o}m projection method for realizing KAF. In addition to the guaranteed theoretical convergence, we numerically demonstrate the effectiveness of this approach on higher-dimensional problems where the relevant kernel features are difficult to capture with the Nystr\"{o}m method. Given noisy training data, we propose a nonparametric smoother as a de-noising method. Numerically, we show that the proposed smoother is more accurate than EnKF and 4Dvar in de-noising signals corrupted by independent (but not necessarily identically distributed) noise, even if the smoother is constructed using a data set corrupted by white noise. We show skillful prediction using the KAF constructed from the denoised data. 
\end{abstract}

% keywords can be removed

\section{Introduction} \label{sec:intro}

% supervised learning paradigm
A long-standing issue in the applications of dynamical systems is to predict time series of observables given partial observations. This problem has classically been studied from various angles under different names in the literature (i.e. reduced-order modeling, closure modeling, subgrid parameterization, etc), but more recently it has also been viewed as a machine learning problem. In particular, at the core of this modeling problem is a supervised learning task to find a map that takes appropriate covariate data (an observable in the past and/or present times) to the desired response function (an observable at the future times). When the covariate data is a delay-embedded observable, the target map provides a non-Markovian prediction model. The realizations of this problem with state-of-art machine learning algorithms involving deep/recurrent neural networks have reported superb numerical performances even when the underlying dynamics are highly nonlinear and high-dimensional \cite{vlachas2018data,ma2018model,pan2018data,HJLY:19}. In the context of partially known dynamics, the recent work in \cite{HJLY:19} formulated the target function as a conditional expectation associated with an appropriate probability space and showed that the corresponding supervised learning framework (which is similar to the one proposed in \cite{pan2018data,jh:20}) produces an approximate closure model whose solutions converge (strongly) to those of the underlying dynamics for finite time when both models are initialized with the same initial conditions. Building on this positive result, one of the goals of this paper is to understand the regression problem corresponding to the supervised learning task from the viewpoint of dynamical systems theory and reduced-order modeling.

% Classical theory: Takens and MZ
Due to the classical theory of dynamical systems, this modeling framework is closely related to the delay-embedding theorem \cite{TAKENS2010345} which has served as a foundation for attractor reconstruction from time series. We will argue that when the embedding theorem is satisfied, the regression (or target) function is theoretically consistent with the component of the flow induced by the delay-coordinate map. From the reduced-order modeling viewpoint, the same learning task can be formulated as a problem of deriving, from first principles, a set of effective equations that determines the evolution of the observable (e.g., \cite{chk:02,cho2016numerical,parish2017dynamic,chuli:2019,price2019renormalized}). 

The Mori-Zwanzig (MZ) formalism \cite{Zwanzig61,Mori65} has been proposed by this community as a natural framework for deriving such a set of effective equations for forecasting the time series of partially observed dynamical systems. The appeal of using the MZ formalism is that the resulting system is represented by an equation that involves projected linear evolution operators. In contrast to geometrical state-space approaches, the operator-theoretic approach focuses on the induced linear action of dynamical systems on appropriately chosen spaces of observables despite the nonlinearity of the flow map. In the context of the MZ formalism, this allows one to compartmentalize the contribution of the observable at the present time (the Markovian term), the observable in the past (the memory/non-Markovian term), and the orthogonal dynamics of the trajectory of the observables at the future times with a collection of linear operators. 

While such a representation is attractive for understanding the modeling mechanism, it may not be easily translated into an efficient numerical method. This issue arises due to the fact that the MZ formula states the dependence of the observable at the future time on the entire history of observables and the initial condition. Besides, it is usually difficult to specify the memory kernel as it requires the solution of the high-dimensional orthogonal dynamics \cite{chk:02,gouasmi2017priori}.  Ultimately, the desired computational objective is to have a finite memory approximation.  This issue has given rise to many parametric approximations of the memory kernel, such as the delta function approximation \cite{hijon2010mori}, Krylov subspace approximation \cite{LiXian2014}, series expansion \cite{LiStinis2019,zhu2018faber}, and rational approximation \cite{hl:15}, just to name a few. While these approaches have shown positive results when addressing specific applications, they either require the knowledge of the full model and/or they are subjected to modeling error when the memory kernel is not adequately represented by the specified parametric model. We will argue that if the hypotheses of the delay-embedding theorems are satisfied, the representation of the MZ equation with the projection operator obtained through the corresponding regression framework can be simplified to a computationally tractable model. In particular, the MZ equation consists of only the ``Markovian'' term associated with the delay-embedded sequences, which is exactly the regression function given by the supervised learning framework. The connection between supervised learning, delay-embedding theory, and the MZ formalism suggests that the regression framework is indeed a natural approach for predicting time series of partially observed dynamics.
We should point out that this connection partially explains the empirical successes reported in
 \cite{vlachas2018data,ma2018model,pan2018data,HJLY:19} since they all adopted this regression modeling paradigm. One unexplained component of these empirical successes is the consistency of their estimators.  
In these papers, the authors approximated the target function using a neural network model (which is in the form of a composition of activation functions) which depends nonlinearly on possibly a very large number of parameters (depending on the depth and width of the neural network architecture). Thus, the training phase often involves a nonlinear, highly non-convex, optimization problem, and finding the global optimizer for such a problem can be a difficult task given that most solvers convergence is guaranteed locally. While this is an interesting direction, we will not explore it here. 
In this paper, we study a class of linear estimators that can be translated into computational algorithms with theoretical guarantees. In particular, we consider the kernel analog forecast (KAF) which has found applications in finance \cite{tay2001application} and climate sciences \cite{zhao2016analog,alexander2017kernel,comeau2017data,comeau2019predicting}. KAF is a kernel regression method designed for the purpose of predicting time series generated by an observable of a dynamical system. The term ``analog'' refers to the fact that KAF is a generalization of the classical analog forecasting method proposed by Lorenz \cite{lorenz1969atmospheric}, for which the prediction is determined based on the affinity of the present states and the historical analog. In this context, the so-called ``kernel trick'' allows one to identify the analogs (feature space) with an appropriately chosen kernel. This, in turn, allows one to access an estimator that lies in a Reproducing Kernel Hilbert Space (RKHS) induced by the associated kernel features, with universal approximation properties. A key advantage of the RKHS formulation is that properties of the elements of the space are inherited by corresponding properties of the kernel. In particular, if the kernel is bounded, then functions in the RKHS are also bounded. Likewise, functions in an RKHS inherit the regularity of the kernel. This important property allows one to establish uniform convergence of the estimator, which justifies the use of KAF as an interpolator. In the context of dynamical systems forecasting, the natural function space (e.g., an $L^2$ space associated with an invariant measure) is usually not known explicitly, yet relationships between kernel integral operators and RKHSs allow one to empirically access the subspace of $L^2$ through a set of orthogonal basis functions corresponding to ordered eigenvalues. In this case, there is a natural mapping of the $L^2 $ basis vectors corresponding to nonzero eigenvalues to orthogonal RKHS functions, and, under appropriate positivity conditions on the kernel, the latter span a dense subspace of the corresponding $L^2$ space. With orthogonality at hand, one can control the accuracy of the estimate by a finite eigenbasis representation and, simultaneously, avoid the large matrix inversion problem with the radial-type kernels. Finally, the RKHS structure allows one to evaluate the estimator on new data points using a classical interpolator, the Nystr\"om projection method. It should be noted that this construction does not require that the covariate time series is Markovian, and is therefore well suited to forecasting under partial observations; e.g., see  \cite{BurovEtAl20} for applications of KAF to prediction of slow components of multiscale systems exhibiting averaging or homogenization. 

% kernel smoothing
While KAF is theoretically sound \cite{alex2019operatortheoretic}, it may face practical limitations, especially when both the covariate space and the support of the pushforward of the invariant measure on the covariate space are high dimensional. This issue is mainly due to lack of guarantees that the leading eigenfunctions induced by a generic kernel on a high-dimensional covariate space adequately capture the response (predictand) variable of interest. To alleviate this limitation, while also reducing computational complexity, we propose to realize KAF with a \emph{kernel smoothing} technique, whose basic idea is to apply a discrete convolution of a Markov operator on the response functions. We show that the proposed kernel smoothing method is a consistent estimator of the optimal regression function, i.e., the conditional expectation of the response given the covariate data. Using the variable-bandwidth kernels introduced in \cite{bh:15vb}, we numerically demonstrate the effectiveness of kernel smoothing compared to the Nystr\"om method in estimating the full discrete MZ equation in situations where the covariate space is relatively high-dimensional. On the other hand, when the covariate space is low dimensional, the Nystr\"om method is generally a better choice since the response variable is more likely to be well represented by the leading empirical kernel features. 

% noisy data
Another critical issue that often arises in practical applications is that the available observables are subjected to noises (of possibly unknown nature). This poses a question in the accuracy of the KAF estimators since the noises in the response and covariate data may yield an ill-posed regression problem. In this paper, we propose a \emph{non-parametric smoother}, constructed using the Nystr\"om projection method, to denoise observables corrupted by independent (but not necessarily identically distributed) noises. In our applications, we will show the effectiveness of the proposed smoother in denoising signals corrupted by various noise types, including time varying noise, even if the smoother is constructed using a data set corrupted by independent and identically distributed  (i.i.d.) Gaussian noise. From our numerical tests, we will find that the proposed smoother produces more accurate estimates than two popular data assimilation methods that are presently used in operational weather forecasts: the Ensemble Kalman filter \cite{evensen:94} and the 4D-Variational approach \cite{lorenc:86}, both of which require the true governing equations of the observed components. Using the smoothed data, we numerically verify that the kernel smoothing method is effective in predicting the response variable. We will show that this blended ``projection-smoothing" approach is able to produce a reasonably accurate prediction from purely noisy observables.

% organization of the paper
This paper is organized as follows. In Section~\ref{sec:regression}, we review the kernel-based regression framework for supervised learning tasks. In Section~\ref{NystromEst}, we discuss the Nystr\"om projection method. While the presentation follows closely that in \cite{alex2019operatortheoretic}, in the current discussion, we do not present the regression problem for time series generated by ergodic dynamical systems and only describe it on i.i.d.\ training data. We complete the discussion in Section~\ref{NystromEst} with a simple statistical error bound. In section~\ref{sec::kernelsmooth}, we present the kernel smoothing method, and prove its consistency and associated error bounds using variable bandwidth kernels \cite{bh:15vb}. In Section~\ref{sec:prediction}, we discuss the problem of predicting observables of time series generated by dynamical systems. Since the only available training data is the time series of the relevant observables, we briefly review the discrete MZ formalism for reduced-order modeling in Section~\ref{sec31}. In Section~\ref{sec32}, we focus on estimating the solution operator of the projected discrete MZ equation with the KAF estimator. We demonstrate the performance of the estimator on a Hamiltonian system and the five-dimensional chaotic Lorenz-96 dynamical system. In Section~\ref{sec:estMZ}, we discuss the connection of the proposed nonparametric regression framework with the delay-embedding and MZ formalism. In particular, we will show that if the hypothesis in the delay-embedding theory is satisfied, the regression function is indeed a component of the flow map. Furthermore, the MZ equation derived using the projection operator obtained by the regression framework consists of only the ``Markovian'' term and it is exactly represented by the corresponding regression function. Supporting numerical examples on the two same dynamical systems are given. In Section~\ref{sec4}, we consider data corrupted by independently distributed noises. A non-parametric smoother based on the Nystr\"om projection method is presented as a denoising method in Section~\ref{subsection41}. Subsequently, in Section~\ref{subsection42}, we numerically verify the prediction skill of the KAF estimator when it is trained using the smoothed data. In Section~\ref{sec5}, we close this paper with a summary and outlook of open problems.

\section{Nonparametric regression}\label{sec:regression}

Given spaces $\mathcal{X}$ and $\mathcal{Y}$, a basic problem of supervised learning is to construct a map $F:\mathcal{X}\to\mathcal{Y}$  from samples of labeled data, $\{(x_i,y_i) \in \mathcal{X}\times\mathcal{Y} \}_{i=1,\ldots,N}$, such that $F(x_i)$ optimally approximates $y_i $ in a suitable sense. Here, we require that $\mathcal{Y}$ be a Hilbert space so that we can apply orthogonal projections, as well as compute expectations and other statistical functionals. On the other hand, we allow $\mathcal{X}$ to be nonlinear.  In order for the  target function $F$ to be predictive, we relate $x_i$ and $y_i$ by assuming that they are realizations of random variables $X$ and $Y$ with common domain $\Omega$. We assume that $\Omega$ is a probability space equipped with a $\sigma$-algebra $\mathcal{B}(\Omega)$ and probability measure $\mu$. We call $\mathcal{X}$ the covariate space and $\mathcal{Y}$ the response space. The corresponding maps $X$ and $Y$ are called the covariate map and the response map, respectively. 

Consider the Hilbert spaces $H=\big\{f:\Omega\rightarrow \mathcal{Y}\mid \int_{\Omega} f^2(\omega) \ d\mu(\omega) <\infty\big\}$, $V=\{g:\mathcal{X}\to  \mathcal{Y}: g\circ X\in H\}$, and $H_X=\{f\in H: f = g \circ X \mbox{ for some } g\in V\}$. Note that $H=L^2(\mu)$ and $V=L^2(\nu)$ where $\nu= X_*\mu$ is the pushforward of $\mu$ via $X$. Moreover, $H_X$ is the Hilbert subspace of $L^2(\mu)$ that contains equivalence classes  of square-integrable functions which are measurable with respect to the $\sigma$-algebra generated by $X$. In general, there are many ways to construct a predictive map $F$.  The least-squares approach is to construct an $F$ that minimizes the mean square error. A standard result from statistics is that this estimator is given by the regression function, which is also known as the conditional expectation function. That is,
\BEA
\mathbb{E}[Y|\cdot] = F := \arg \min_{g\in V} \| Y - g\circ X \|^2_{H}, \label{optimal}
\EEA
where the conditional expectation $\mathbb{E}[\cdot |X]$ can be seen as an orthogonal projection of $H$ onto $H_X$. In this paper, we will denote the orthogonal projection $P:H\to H_X$ as the conditional expectation $\mathbb{E}[\cdot |X]$. In appropriate context, we will also use $P:H\to S_X\subseteq H_X$, to denote an arbitrary orthogonal projection onto its range space, $S_X = \mbox{ran}(P)$ such that $S_X^\perp = \mbox{null} (P)$ and $H=S_X\oplus S_X^\perp$, where the orthogonality is defined with respect to the inner product of $H$.

When $X$ is not injective, as in many applications, one cannot approximate the response $Y\in H$ to arbitrary precision by elements of $H_X$. However, one can still construct an optimal estimator of $Y$ using the target function $F\in V$. In the remainder of this section, we discuss two methods for estimating $\mathbb{E}[Y|\cdot]$ from samples of labeled data $\{(x_i,y_i) \in \mathcal{X}\times\mathcal{Y}, \}_{i=1,\ldots,N}$. The first one is the Nystr\"om method which is an interpolation of an eigenbasis representation of the estimator. The second method is the kernel smoothing that employs a convolution operation associated with a Markov kernel. For the remainder of this section, we restrict our discussion to real-valued functions, so that $\mathcal{Y}=\mathbb{R}$ and $H = \{f:\Omega\rightarrow \mathbb{R}\mid \int_{\Omega} f^2(\omega) \ d\mu(\omega) <\infty\}$. Since our applications involve $\mathcal{Y}=\mathbb{R}^n$, a componentwise generalization to the finite-dimensional vector-valued case is immediate.

\subsection{Nystr\"om method} \label{NystromEst}

If $V$ is equipped with an orthonormal basis $\{u_j\}_{j\in \mathbb{N}}$ and $X$ is injective, then $\{\phi_j = u_j\circ X\}_{j\in \mathbb{N}}$ forms an orthonormal basis of $H$. In this case, any $Y\in H$ can be arbitrarily estimated, in $H$-norm, by 
\BEA
\mathbb{E}_L[{Y|X}]: = \sum_{j=0}^{L} \langle Y, \phi_j \rangle_{H} \phi_j,   \label{finitesum}
\EEA
up to any desirable precision by taking $L\to\infty$. Due to the properties of orthogonal projection, the estimator 
\BEA
\mathbb{E}_L[Y|\cdot] = \sum_{j=0}^{L} \langle Y, \phi_j \rangle_{H} u_j, \label{estimator}
\EEA
is an optimal estimator from $\mbox{span}\{\phi_0,\ldots,\phi_L\}\subset H$. As mentioned above, when $X$ is not injective, $\mbox{span}\mbox\{\phi_j\}_{j\in \mathbb{N}}\subsetneq H$ so one cannot recover arbitrary target functions $Y\in H$. However, $\mathbb{E}_L[Y| \cdot]$  is a consistent estimator of $\mathbb{E}[Y|\cdot]\in V$ so that $\lim_{L\to\infty}\mathbb{E}_L[Y|\cdot] = \mathbb{E}[Y|\cdot]$ in $V$. 

A practical issue in employing the estimator \eqref{estimator} is that orthonormal bases  of $H$ as well as $V$ are not available. The whole point of nonparametric regression is to construct an estimator for $\{\phi_0,\phi_1,\ldots\}$ from the random samples of observables $\{x_i:i=1,\ldots, N\}$, where $x_i=X(\omega_i)$ are realizations of the covariate map $X$.  Kernel-based algorithms \cite{cl:06,bh:15vb} are often used to obtain the function value $u_j(x_i) = u_j\circ X(\omega_i) =\phi_j(\omega_i)$, which can subsequently be used to estimate the inner product in \eqref{estimator}. For our purposes, we also need to evaluate the estimator in \eqref{estimator} on new covariate data that do not lie in the (finite) training data set. This evaluation can be done using an interpolation scheme such as the Nystr\"om method that extends $u_j$ on new covariate data disjoint from the finite sample of observations. To justify the validity of such an interpolation method, uniform convergence of the estimator is usually required rather than $V$-norm convergence. 

One way to ensure uniform convergence is to construct an estimator in a reproducing kernel Hilbert space (RKHS) $\mathcal{H}$ of continuous functions such that $\mathcal{H}$ is dense in $H_X$.  In particular, let $k:\Omega\times\Omega\to\mathbb{R}$ be the pullback of a kernel $\kappa:\mathcal{X}\times \mathcal{X}\to\mathbb{R}$ on the covariate space. That is, $k$ is symmetric positive definite and $k(\omega,\omega') = \kappa(X(\omega),X(\omega')))$. By the Moore-Aronszajn theorem, there exists a unique Hilbert space $\mathcal{H}$ (the RKHS),  of real valued functions $f:\Omega\to \mathbb{R}$ with the reproducing property:  $\mathcal{H}=\overline{\mbox{span}\{k(\omega,\cdot), \forall\omega\in\Omega\}}$ and every $f\in\mathcal{H}$ and $\omega\in \Omega$ satisfies $f(\omega)=\langle k(\omega,\cdot),f\rangle_{\mathcal{H}}$. Since the kernel $k$ is a pullback kernel of $\kappa$,  every function $f\in\mathcal{H}$ can be expressed as $f=g\circ X$ for some continuous function $g:\mathcal{X}\to\mathbb{R}$. If  $\Omega$ is compact and $k$ is continuous, one can show that $\mathcal{H}$-norm convergence implies uniform convergence so that  $\mathcal{H}\subset C(\Omega)$. For non-compact domains, a bounded kernel ensures that $\mathcal{H} \subset C_b(\Omega)$ \cite{christmann2008support}. 

While it is convenient to represent functions in $\mathcal{H}$ as a linear superposition of kernel sections, namely, $f = \sum_{i=1}^\infty a_i k(\omega_i,\cdot)$ with $\omega_i\in\Omega$, empirical representations involve a partial summation of $N$ terms, where $N$ denotes the number of training samples. For large datasets, as in our applications, specification of the coefficients $a_i$ involves an inversion of a large matrix and the repetitive function evaluation is numerically expensive. If a radial-type kernel is chosen, as in many applications, then we arrive at the at the so-called kernel ridge regression or radial basis function interpolation, depending on the literature. The estimator in \eqref{estimator} is  proposed as an alternative to avoid this computational issue by leveraging the inner product structure of $H$. To that end, consider the reproducing kernel $k$ from the perspective of an integral operator $K_\mu: H \to \mathcal{H}$ defined as
\BEA
K_\mu f = \int_\Omega k(\cdot,\omega) f(\omega) d\mu(\omega), \label{integraloperator}
\EEA
where $\mu$ is assumed to be compactly supported on $M\subset \Omega$. This is a compact operator with adjoint $K_\mu^*:\mathcal{H}\to H$ that is also compact. By the spectral theorem, the compact, self-adjoint and positive-definite integral operator $G_\mu:=K_\mu^*K_\mu : H\to H$ has eigenvalues $ \lambda_0 \geq \lambda_1 \geq \cdots \searrow 0^+$ so that the corresponding eigenfunctions $\{ \phi_0, \phi_1, \ldots \} $ form an orthonormal basis of $H$. In fact, defining, $\psi_j = K_\mu\phi_j/\lambda_j^{1/2}$ for $\lambda_j>0$, we have,
\BEA
\langle \psi_i,\psi_j \rangle_{\mathcal{H}} = \frac{1}{\lambda_i^{1/2}\lambda_j^{1/2}} \langle K_\mu \phi_i,K_\mu \phi_j \rangle_{\mathcal{H}} = \frac{1}{\lambda_i^{1/2}\lambda_j^{1/2}} \langle K^*_\mu K_\mu \phi_i, \phi_j \rangle_{H} = \frac{\lambda_i^{1/2}}{\lambda_j^{1/2}} \langle \phi_i, \phi_j \rangle_{H} = \delta_{ij},\nonumber
\EEA
which means that $\{\psi_0,\psi_1,\ldots\}$ is an orthonormal set in $\mathcal{H}$. By Mercer's theorem, we have an explicit representation $k(\omega,\omega') = \sum_{j=0}^\infty \lambda_j\varphi_j(\omega)\varphi_j(\omega')=
\sum_{j=0}^\infty\psi_j(\omega)\psi_j(\omega')$, converging uniformly for $ (\omega,\omega') \in M \times M$, where $\varphi_j = \lambda_j^{-1/2}\psi_j $ denotes the continuous representative of eigenfunction $\phi_j$. 
The so-called ``kernel trick'' specifies an explicit choice of kernel $k$, such as the Gaussian kernel, to avoid computing the $\ell_2$ inner-product between feature vectors $(\psi_0(\omega),\psi_1(\omega),\ldots)$ and $(\psi_0(\omega'),\psi_1(\omega'),\ldots)$. Our perspective is to rely on the orthogonality of the eigenbasis to approximate  the target function of interest  through the representation in \eqref{estimator} and use the RKHS theory to establish the convergence of the estimator as $L\to\infty$. 

One of the most important aspects of the integral operator $K_\mu$ is that we can define an interpolation  (Nystr\"{o}m) operator $\mathcal{N}_\mu: D(\mathcal{N}_\mu)\to \mathcal{H}$ as $\mathcal{N}_\mu\phi_j := \psi_j/\lambda_j^{1/2}= K_\mu \phi_j/\lambda_j:=\varphi_j$, whose domain $D(\mathcal{N}_\mu) = \{f =\sum c_k\phi_k \in  H| \sum_k c_k^2/\lambda_k<\infty\}$ contains functions of higher regularity than arbitrary elements of $H$. Note that if $D(\mathcal N_\mu)$ is equipped with the norm $ \lVert f \rVert^2 = \sum_k c_k^2/\lambda_k$, then it is isometrically isomorphic to $\mathcal H(M)$, the restriction of $ \mathcal H$ to the support $M$. Notice that the operator $\mathcal{N}_\mu$ maps the eigenfunction $\phi_j\in D(\mathcal{N}_\mu)$ to the continuous function $\varphi_j$. As a result,  $f\in D(\mathcal{N}_\mu)$  has a continuous representation $\mathcal{N}_\mu f = \sum_j c_j \mathcal{N}_\mu \phi_j = \sum_{j} c_j \varphi_j$. Moreover, the map $K_\mu^*$ is a left inverse of $\mathcal{N}_\mu$ since
\BEA
K_\mu^*\mathcal{N}_\mu f = \sum_{j} c_j K_\mu^* \varphi_j = \sum_{j} c_j K_\mu^* K_\mu \frac{\phi_j}{\lambda_j} = \sum_{j} c_j\phi_j = f.\label{identitymap}
\EEA
This means that the map $K_\mu^*\mathcal{N}_\mu: D(\mathcal{N}_\mu)\to \overline{\mbox{ran}K_\mu}$ identifies functions in $D(\mathcal{N}_\mu)$ with their continuous representation in $\mathcal{H}$ through the Nystr\"{o}m operator, as a function in $\overline{\mbox{ran}K_\mu}$. 

In our case, the target function is $\mathbb{E}[Y|\cdot]\in V$ or $\mathbb{E}[Y|X]\in H_X$. Thus we can consider the operator \eqref{integraloperator} but with domain $H_X$. In this case,  an orthonormal set of continuous functions $\{\psi_0,\ldots,\psi_L\}$ in $\mathcal{H}$ satisfies $\psi_j = u_j \circ X$ for some continuous functions $\{u_0,\ldots,u_L\}$ that can be approximated from the covariate data. Using this basis, for each $\mathbb{E}_L[Y|X]\in D(\mathcal{N}_\mu)$, one can build an estimator for $\mathcal{N}_\mu \mathbb{E}_L[Y|X] \in \mathcal{H}$ which can be represented as
\BEA
\mathcal{N}_\mu \mathbb{E}_L[Y|\cdot ] =\sum_{j=0}^L \langle Y, \phi_j \rangle_{H} \frac{u_j}{\lambda_j^{1/2}}.\label{nystromEm}
\EEA
It is important to note that if the reproducing kernel $k$ of the RKHS $\mathcal{H}$ is a pullback of a strictly positive definite kernel $\kappa:\mathcal{X}\times\mathcal{X}\to\mathbb{R}$, then the domain $D(\mathcal{N}_\mu)$ is dense in $H_X$. To see this, take any function $f \in H_X$ and, since $\mbox{span}\{\phi_0,\phi_1,\ldots\}$ is dense in $H_X$, we have that $f=\sum_k c_k\phi_k$, where each eigenfunction is associated with a strictly positive eigenvalue. Furthermore,
\BEA
\sum_{k=0}^\infty \frac{{c_k}^2}{\lambda_k} = \sum_{k=0}^\infty \frac{{c_k}^2}{\lambda_k} \langle \phi_k,\phi_k \rangle_{H_X}  = \sum_{k=0}^\infty \frac{{c_k}^2}{\lambda_k} \langle K_\mu \phi_k,K_\mu \phi_k \rangle_{\mathcal{H}} =  \sum_{k=0}^\infty{c_k}^2 \langle  \psi_k, \psi_k \rangle_{\mathcal{H}}= \sum_{k=0}^\infty {c_k}^2 = \lVert f \rVert_{H_X}^2 <\infty,\nonumber
\EEA
and we conclude that any function $f\in H_X$ can be approximated by a function in $D(\mathcal{N}_\mu)$ at arbitrary precision. From \eqref{identitymap}, one can see that the operator $K_\mu^*\mathcal{N}_\mu: D(\mathcal{N}_\mu) \to D(\mathcal{N}_\mu)$ is an identity map (a bounded operator). By the bounded linear transformation theorem, the closed extension of $K_\mu^*\mathcal{N}_\mu$ is the identity map on $\overline{D(\mathcal{N}_\mu)}=H_X$. This means that any function in $H_X$ can be approximated to arbitrary precision in $H$-norm by a function in $K^*_\mu\mathcal{H} = D(\mathcal N_\mu)$. In particular, as $L\to\infty$, the estimator in \eqref{nystromEm} converges to the target function in $H$-norm, i.e., $\lim_{L\to\infty}K^*_\mu\mathcal{N}_\mu \mathbb{E}_L[Y|X] = \mathbb{E}[Y|X]$. If it now happens that $\mathbb{E}[Y|X] $ has a representative in $\mathcal H$, then the estimator converges to that representative in $\mathcal H$-norm, and thus uniformly, on the support of $\mu$.

As mentioned above, in practice, we have no access to the basis functions $\{\phi_0,\ldots,\phi_L\}$ or $\{u_0,\ldots,u_L\}$. Given the pairs of labeled data points $\{ (x_i,y_i )\}_{i=1,\ldots,N}$, where $x_i$ are i.i.d.~samples of $X$, we first describe an empirical estimate of $\phi_j(\omega_i)=u_j(x_i)$. Let $G_{\mu_N}:=K_{\mu_N}^*K_{\mu_N}$, where $K_{\mu_N} : H_N \to \mathcal H$ and $K^*_{\mu_N} : \mathcal H \to H_N$ are defined as in \eqref{integraloperator} and the corresponding adjoint with $H= L^2(\mu_N)$ replaced by $ H_N := L^2(\mu_N) $. Here, $\mu_N = \sum_{j=1}^N \delta_{\omega_j} / N $ is the discrete sampling measure, and $L^2(\mu_N)$ the corresponding finite-dimensional Hilbert space equipped with the inner product $\langle f,g \rangle_{H_N}=\frac{1}{N}\sum_{i=1}^N f(\omega_i)g(\omega_i)$. For simplicity of exposition, we will assume that all sampled states $ \omega_i $ are distinct, so $H_N$ is an $N$-dimensional Hilbert space, isomorphic to $\mathbb R^N $ equipped with a normalized dot product. In that case, the operator $G_{\mu_N}$ is represented by an $N\times N$ kernel matrix $\mathbf{G}_N =[\langle e_{i,N},G_{\mu_N}e_{j,N}\rangle_{H_N}] =[\kappa(x_i,x_j)]$, where $e_{j,N} $ are the standard orthonormal basis vectors of $H_N$ with $e_{j,N}(\omega_i) = N^{1/2}\delta_{ij}$.

Let $\{\lambda_{j,N},\phi_{j,N}\}$ be the $j$th eigenvalue and eigenvectors of $\mathbf{G}_N$, respectively. It is well known that the approximation of $G$ by $G_n$ is spectrally consistent \cite{VonLuxburgEtAl08}; that is the sequence of eigenvalues $\lambda_{j,N} \to \lambda_j$, as $N\to\infty$. 
Moreover, the continuous representative, $\mathcal{N}_{\mu_N}\phi_{j,N} = \psi_{j,N}/\lambda_{j,N}^{1/2}$ converges to $\mathcal{N}_{\mu}\phi_{j} = \psi_j/\lambda_j^{1/2}$ as $N\to \infty$ in $\mathcal{H}$.
%We should point out that $\psi_{j,N}=u_{j,N}\circ X$, where $u_{j,N}\in C(\mathcal{X})$. 
Denoting $\vec{y}=(y_1,\ldots,y_N)^\top\in\mathbb{R}^N$, we have
\BEA
\langle \vec{y},\phi_{j,N} \rangle_{L^2(\mu_N)} = \frac{1}{N} \sum_{i=1}^N {y_i} \phi_{j,N}(\omega_i) = \int_{\Omega} Y(\omega)\mathcal{N}_{\mu_N} \phi_{j,N}(\omega) d\mu_N(\omega)\longrightarrow \int_{\Omega} Y(\omega) \mathcal{N}_{\mu} \phi_{j}(\omega)d\mu(\omega) = \langle Y,\phi_j \rangle_{H},\nonumber
\EEA
as $N\to\infty$, where we have used the law of large numbers for i.i.d. samples. For each $j$, 
\BEA
\mathbb{E}_\mu  [\langle \vec{y},\phi_{j,N}\rangle_{L^2(\mu_N)} ]= \frac{1}{N} \sum_{i=1}^N \mathbb{E}_\mu [{Y}\mathcal{N}_{\mu_N} \phi_{j,N}] = \mathbb{E}_\mu [{Y} \mathcal{N}_{\mu_N}\phi_{j,N}] =\mathbb{E}_\mu [{Y} \phi_{j}] +\mathcal{O}(\delta),\label{bias}
\EEA
where $\delta$ is an error bound of the eigenfunction estimation. In the proposition below, we will specify $\delta$ on a manifold without boundary based on the $L^2$ result from \cite{GTetal2019}.
The standard Monte-Carlo error suggests that
\BEA
\mathbb{E}_\mu  \Big[\big(\langle \vec{y},\phi_{j,N}\rangle_{L^2(\mu_N)} -\mathbb{E}_\mu  [Y\mathcal{N}_{\mu_N}\phi_{j,N} ]\big)^2\Big] = \frac{1}{N} \mathbb{E}_\mu  \Big[ (Y\mathcal{N}_{\mu_N}\phi_{j,N} -\mathbb{E}_\mu  [Y\mathcal{N}_{\mu_N}\phi_{j,N} ]\big)^2\Big] = \frac{\mbox{Var}[Y\mathcal{N}_{\mu_N}\phi_{j,N}]}{N}.\nonumber
\EEA
Without loss of generality, suppose that $\mathbb{E}_\mu[Y] = \mathbb{E}_\mu[\phi_{j}]=0$. If $Y$ is continuous on $M\subset \Omega$, the compact support of $\mu$, then
\BEA
\mbox{Var}[Y\mathcal{N}_{\mu_N}\phi_{j,N}] &=& \mathbb{E}_\mu[Y^2(\mathcal{N}_{\mu_N}\phi_{j,N})^2] \leq \|Y^2\|_{\infty}\mathbb{E}_\mu[(\mathcal{N}_{\mu_N}\phi_{j,N})^2]\nonumber \\ &\leq&  \|Y^2\|_{\infty}(\mathbb{E}_\mu[\phi_{j}^2] + \mathcal{O}(\delta^2)) = \|Y^2\|_{\infty}(1 + \mathcal{O}(\delta^2)),\nonumber
\EEA
where we have used the H\"older inequality and the orthonormality of $\phi_j$. Together with \eqref{bias}, we have 
\BEA
\mathbb{E}_\mu  \Big[\big(\langle \vec{y},\phi_{j,N}\rangle_{L^2(\mu_N)} -\mathbb{E}_\mu  [Y\phi_{j} ]\big)^2\Big] \leq C\left(\frac{1}{N}+ \frac{\delta^2}{N} + \delta^2\right),\label{MCerror}
\EEA
for some constant $C>0$. 

\comment{\color{blue}In fact, an application of the Cauchy Schwarz inequality along with the orthonormality of $\phi_j$ gives the error estimate
\BEA
\mathbb{E}_\mu \Big[\big(\frac{1}{N} \sum_{i=1}^N {y_i} \phi_{j,N}(\omega_i) - \langle Y,\phi_j \rangle_{H}\big)^2\Big]\leq  \frac{1}{N}\mathbb{E}_\mu[Y^2]\mathbb{E}_\mu[\phi_j^2] - \langle Y,\phi_j \rangle_{H}^2 \leq \frac{\mathbb{E}_\mu[Y^2]}{N}.
\EEA}
We should point out that if the samples $\{\omega_i\}$ form a time series generated by an ergodic and stationary dynamical system, then the convergence can still be achieved via the Birkhoff ergodic theorem, but the convergence rate would depend on the mixing rate of the underlying processes \cite{Davydov:68,Hang:14}. Together with the convergence of the continuous representative, we can conclude that the discrete estimator
\BEA
\mathbb{E}_{L,N} [Y|X] := \sum_{j=0}^L \langle \vec{y},\phi_{j,N} \rangle_{L^2(\mu_N)} \phi_{j,N}, \nonumber
\EEA
has a continuous representative
 \BEA\mathcal{N}_{\mu_N} \mathbb{E}_{L,N} [Y|X] =  \sum_{j=0}^L \langle \vec{y},\phi_{j,N} \rangle_{L^2(\mu_N)} \psi_{j,N}/\lambda_{j,N} \label{Nystrom}\EEA 
 that converges in $\mathcal{H}$-norm to $\mathcal{N}_{\mu} \mathbb{E}_{L} [Y|X]$ as $N\to\infty$. Also, the left pseudo-inverse, $K_{\mu}^* \mathcal{N}_{\mu_N} \mathbb{E}_{L,N} [Y|X] \to K_{\mu}^* \mathcal{N}_{\mu} \mathbb{E}_{L} [Y|X]= \mathbb{E}_{L} [Y|X]$ as $N\to\infty$ in $H_X$. Taking $L\to\infty$ after $N\to\infty$, we establish the consistency of the estimator with the target function, $\mathcal{N}_{\mu_N} \mathbb{E}_{L,N} [Y|X] \to \mathbb{E} [Y|X]\in H_X$. 
 
 Let $ \nu_N = \mu_N \circ X^{-1}$ be the pushforward of the sampling measure on covariate space $\mathcal X$. Computationally, we can estimate the discrete orthonormal basis $\{u_{0,N},u_{1,N}, \ldots,u_{L,N}\}$ with respect to $L^2(\nu_N)$ by solving an eigenvalue problem associated with a Markov operator $G_{\mu_N,\epsilon}$ constructed using a  decreasing kernel $k_{\epsilon}$ defined with bandwidth parameter $\epsilon$ (see also remark $3$ of \cite{GTetal2019}).  Note that the pullback is given as $\phi_{j,N,\epsilon} = u_{j,N,\epsilon} \circ X$. If $\mathcal{X}$ is a $d$-dimensional compact smooth manifold embedded in $\mathbb{R}^n$, then $\mathcal{N}_{\mu_N}u_{j,N,\epsilon}$ converges to the eigenfunctions $u_j$ of the Laplace-Beltrami operator (positive definite with  respect to $V$) as $N\rightarrow \infty$ and $\epsilon\rightarrow 0$. If $\mathcal \nu_N$  has a smooth density with respect to the volume form, then the Laplace-Beltrami is defined with a conformally changed Riemannian metric inherited by $\mathcal{X}$ from the ambient space $\mathbb{R}^n$. In this case, we have:

\begin{prop}
Let $\mathcal{X}$ be a $d$-dimensional compact smooth Riemannian manifold with no boundary. Let $Y := F\circ X$ such that $F:\mathcal{X}\to\mathcal{Y}$ belongs to a Sobolev class, $ H^\beta(\mathcal{X}):=\{F\in V | \hat{F}_{j}:= \langle F,u_j \rangle_{V}, \sum_j  \zeta_j^\beta \hat{F}_{j}^2 <\infty, \beta>0  \}$, where $\zeta_j$ is the eigenvalue of the Laplace-Beltrami operator associated with eigenfunction $u_j$, approximated with $u_{j,N,\epsilon}$ as discussed in the preceding paragraph. Assume also that $Y\in C(M)$, where $M\subset\Omega$ denotes the compact support of the invariant measure $\mu$.
Then, with $\mu_N=\sum_{j=1}^N\delta_{\omega_j}/N$ and ${\mathcal{N}_\mu} \mathbb{E}_{L,N} [Y|\cdot]$ defined as in \eqref{Nystrom}, we have
\BEA
\mathbb{E}_\nu \Big[(\mathcal{N}_{\mu_N} \mathbb{E}_{L,N} [Y|\cdot] -  \mathbb{E} [Y|\cdot])^2\Big] = \mathcal{O}(LN^{-1},\log(N)^{p_d}N^{-\frac{1}{d}},L^{-\frac{2\beta}{d}}), \nonumber
\EEA
where $p_d = 3/4$ for $d=2$ and $p_d=1/d$ for $d\geq 3$.
\end{prop} 
\begin{proof} 
To compute the error rate, we split the error into the variance error term that arises due to discrete data and the bias term that arises due to the truncation of eigenfunctions:
\BEA
\mathbb{E}_\nu \Big[(\mathcal{N}_{\mu_N} \mathbb{E}_{L,N} [Y|\cdot] -  \mathbb{E} [Y|\cdot])^2\Big] &\leq& \mathbb{E}_\nu \Big[( \mathcal{N}_{\mu_N} \mathbb{E}_{L,N} [Y|\cdot] - K_{\mu}^* \mathcal{N}_{\mu} \mathbb{E}_L [Y|\cdot])^2\Big] + \mathbb{E}_\nu \Big[(K_{\mu}^* \mathcal{N}_{\mu} \mathbb{E}_L [Y|\cdot] - \mathbb{E} [Y|\cdot])^2\Big]\nonumber \\
&\leq& \mathbb{E}_\nu \Big[ \Big(\sum_{j=0}^L \big(\langle \vec{y},\phi_{j,N,\epsilon} \rangle_{L^2(\mu_N)} - \langle Y,\phi_j \rangle_{H}\big)   \mathcal{N}_{\mu_N} u_{j,N,\epsilon} \Big)^2\Big] \ldots \nonumber\\
&+& \mathbb{E}_\nu \Big[ \Big(\sum_{j=0}^L  \langle Y,\phi_j \rangle_{H} ( \mathcal{N}_{\mu_N}u_{j,N,\epsilon}-K_{\mu}^* \mathcal{N}_{\mu}u_{j}  \Big)^2\Big]  + \mathbb{E}_\nu \Big[\big(\sum_{j >L} \langle Y,\phi_j \rangle_{H} u_j \big)^2\Big] \nonumber\\
&\leq& \frac{L}{N} \mathbb{E}_\mu[Y^2] + C\Big( \frac{\log(N)^{p_d}}{N^{\frac{1}{d}}} \Big)+ \sum_{j >L} \langle Y,\phi_j \rangle^2_{H},\nonumber
\EEA 
for some constant $C$ that is independent of $\epsilon, N, d$ but can depend on $L$. In the second equality above for the variance term, we isolate the errors due to Monte-Carlo approximation of the expansion coefficients (which is computed in \eqref{MCerror}), where we suppressed the order $\delta^2/N$ term since it is dominated by the error of order-$\delta^2$ in the discrete approximation of the eigenfunctions Using the recent result in \cite{GTetal2019} for compact manifolds without boundary, the $L^2$-error bound for each eigenfunction (as $\epsilon \rightarrow 0$) is given by $\delta=\mathcal{O}\Big(\frac{\log(N)^{p_d}}{N^{1/d}} \Big)^{1/2}$, where $d$ denotes the intrinsic dimension of $\mathcal{X}$ and $p_d = 3/4$ for $d=2$ and $p_d=1/d$ for $d\geq 3$. 

For all $Y=F\circ X$, we have that $\langle Y,\phi_j \rangle_{H} = \langle F,u_j \rangle_{V} = \hat{F}_j$, and since $F\in H^\beta(\mathcal{X})$, we have
\BEA
 \sum_{j >L} \langle Y,\phi_j \rangle^2_{H} = \sum_{j> L} \hat{F}_j \leq \frac{1}{\zeta_{L+1}^\beta} \sum_{j=0}^\infty \zeta_j^\beta \hat{F}_j^2 \leq C_2 \zeta_{L+1}^{-\beta}, \nonumber
\EEA 
for some constant $C_2>0$. The proof follows by using the Weyl asymptotic estimate for the eigenvalue of the Laplace-Beltrami operator on compact Riemannian manifolds \cite{colbois2013eigenvalues}, $\zeta_{L+1} \sim L^{2/d}$.
\end{proof}

We should point out that balancing the first and last error rates yields the famous minimax optimal rate, $\mathcal{O}( N^{-\frac{2\beta}{2\beta+d}})$ for linear estimators \cite{stone1982optimal}. Thus, unless the response function is highly smooth (e.g, $\beta=d$), such an estimator is subject to the curse of dimension. In practice, the second error rate (corresponding to the estimation of eigenvectors) will dominate in high-dimensional problems even if the target function is smooth.

\subsection{Kernel smoothing estimator}\label{sec::kernelsmooth}

In the previous subsection, we approximated $\mathbb{E}[Y| X]$ with, $\mathbb{E}_{L,N}[Y| X]$, a superposition of eigenvectors of  $\mathbf{G}_N$ and then used Nystr\"om extension \eqref{Nystrom} to evaluate this representation on an out-of-sample point. In this subsection we show that the conditional expectation can also be approximated by an appropriate smoothing function in $H$.

The main idea is motivated by the fact that if $\mathcal X$ is a smooth manifold, any measurable function $g\in V$ can be represented as
\[g(x)=\mathbb{E}_{\delta_x}[g],\]
where $\delta_x$ denotes the Dirac mass centered at $x$. We can then attempt to regularize this integral operation by approximating $\delta_x$ with an appropriate family of Markov kernels that have a smooth density with respect to the pushforward measure $\nu$. 

To that end,  we assume that $\mathcal X = \mathbb R^m$ and the support of $ \nu $ is a smooth, compact $d$-dimensional submanifold $ \mathcal M \subseteq \mathcal X$. We then start with a kernel $S_{\epsilon}:\mathcal{X}\times \mathcal{X} \to\mathbb{R}$, where $\epsilon>0$ is a bandwidth parameter, and perform a sequence of normalizations that yield, asymptotically, the kernel $\kappa_\epsilon$ so that 
\BEA
G_{\epsilon}g(x):=\int_\mathcal{X} \kappa_{\epsilon}(x,x')g(x')d\nu(x')=g(x)+\mathcal{O}(\epsilon), \label{VBDMexpansion}
\EEA
 holds for $g \in V$ and $x\in\mathcal{M}$. The integral operator $G_{\epsilon}$ can then be approximated by a matrix-vector multiplication. In this paper, we use the variable bandwidth construction of the kernel given in \cite{bh:15vb}. This expansion starts with a kernel $S_{\epsilon}$ on $\mathcal{X}\times \mathcal{X}$ of the form 
 \BEA
 S_{\epsilon}(x,x')=\epsilon^{-d/2}\exp\left(-\frac{\|x-x'\|^2}{\epsilon\rho(x)\rho(x')}\right), \nonumber 
 \EEA
 where $\rho>0$ is a bandwidth function  that is chosen to be inversely proportional to a power of the sampling density as in \cite{bh:15vb}. 
 
 For completeness, we describe the construction of the discrete approximation of the operator in \eqref{VBDMexpansion}. Let $x_1,\ldots x_N$ be the observed $m$-dimensional data in $\mathcal{X}$. Then the following steps (which are the diffusion maps normalizations \cite{bh:15vb}) yield a discrete approximation $G_{N,\epsilon}$ of the integral operator $G_\epsilon$, whose discrete representation is denoted by the matrix $\mathbf{G}_{N,\epsilon}=[\langle e_{i,N},G_{N,\epsilon}e_{j,N}\rangle_{L^2(\mu_N)}] = [\kappa_\epsilon(x_i,x_j)]$,
  \BEA
 \begin{aligned}\label{dmnormalization}
 q_{\epsilon}(x_i)&:=\sum_{j=1}^N\frac{S_{\epsilon}(x_i,x_j)}{\rho(x_i)^d}, & S_{\epsilon,\alpha}(x_i,x_j)&:=\frac{S_{\epsilon}(x_i,x_j)}{q_{\epsilon}(x_i)^\alpha q_{\epsilon}(x_j)^{\alpha}},\\
q_{\epsilon,\alpha}(x_i,x_j)&:=\sum_{j=1}^N S_{\epsilon,\alpha}(x_i,x_j) & \mathbf{G}_{N,\epsilon}(x_i,x_j)&:=\frac{S_{\epsilon,\alpha}(x_i,x_j)}{q_{\epsilon,\alpha}(x_i)}.
  \end{aligned}
  \EEA
 The two steps in the first row above are the ``right-normalization'' steps taken to de-bias the possibly non-uniform sampling distribution of the data with a parameter $\alpha$.  In our numerics, we set $\alpha=-d/4$ and $\rho =q_{\epsilon}^{-1/2}$ as in \cite{bh:15vb}. The two ``left-normalization'' steps in the second row of \eqref{dmnormalization} turn $\mathbf{G}_{N,\epsilon}$ into a stochastic matrix. Note that the resulting kernel $\kappa_\epsilon$ from \eqref{dmnormalization} is given in Appendix~A5 of \cite{bh:15vb}. Based on the result in \cite{bh:15vb}, for $\{x_1, \ldots, x_N\} \subset \mathcal M $ with sampling density $q = d\nu/d\text{vol}$, where $\text{vol} $ is the volume form on $\mathcal M$ through its embedding in $\mathcal X$, for fixed $\epsilon$, we have the convergence rate
 \BEA
% \left(\mathbf{G}_{N,\epsilon} \vec{F}\right)_i \longrightarrow G_{\epsilon}F(x_i),\label{VBDMconvergence} 
\left(\mathbf{G}_{N,\epsilon} \vec{g}\right)_i = G_\epsilon g(x_i) + \mathcal{O}\left(\frac{q(x_i)^{1/2+d/4}}{N^{1/2}\epsilon^{2+d/4}},\frac{q(x_i)^{d(d/2-1/4)}}{N^{1/2}\epsilon^{1/2+d/4}}\right),\label{VBDMdiscreteestimate}
 \EEA
as $N\to\infty$.
The second term in the error bound is due to the error in the discrete estimate and the first term is to ensure an order-$\epsilon^2$estimate of $q_\epsilon$.

Our choice of normalization is to ensure an asymptotically unbiased (up to order $\epsilon$) estimate of $g$ in \eqref{VBDMexpansion}. For many applications, it suffices to start with the standard Gaussian kernel with constant bandwidth ($\rho=1$) and apply the steps in the second row in \eqref{dmnormalization} to create a valid transition density.However, in this paper, we will always construct the integral operator in \eqref{VBDMexpansion} using the variable bandwidth kernels due to their accurate estimation of densities in sparsely sampled regions.  To tune the kernel bandwidth parameter $\epsilon$, we use the auto-tuning algorithm in \cite{bh:15vb} which was found to be more effective for variable bandwidth kernels than the Gaussian kernel with a fixed bandwidth ($\rho=1$). Furthermore, we have

\begin{prop}
    Let $P$ be the orthogonal projection of $H$ onto $H_X$. Then for any $g\in V$ and $x_i = X(\omega_i) \in \mathcal{M}$, the relationship
\BEA
\left(\mathbf{G}_{N,\epsilon} \vec{g}\right)_i  = Pg(x_i) + \mathcal{O} \left(\epsilon,\frac{q(x_i)^{1/2+d/4}}{N^{1/2}\epsilon^{2+d/4}},\frac{q(x_i)^{d(d/2-1/4)}}{N^{1/2}\epsilon^{1/2+d/4}}\right),
\EEA
holds in high probability.
\end{prop}  
\begin{proof}
Note that for any $g\in V$, where $f=g\circ X \in H_X$, $\nu = X_*\mu$, and $x_i=X(\omega_i)\in \mathcal{M}$, a change of variables shows that 
 \BEA
G_\epsilon g(x_i) = \int_\mathcal{X} \kappa_{\epsilon}(x_i,x')g(x')d\nu(x') = \int_\Omega \kappa_{\epsilon}(X(\omega_i),X(\omega'))(g\circ X)(\omega')d\mu(\omega') := J_\epsilon f(\omega_i).\label{changeofvariable}
\EEA
 Let $k_{\epsilon,\omega_i}:=k_\epsilon(\omega_i,\cdot)=\kappa_{\epsilon}(X(\omega_i),X(\cdot))$. Since $f \in H_X$, we see that $J_{\epsilon}f(\omega_i)=J_\epsilon P f(\omega_i)$. Thus, for each $\omega_i \in \Omega$,
 \begin{align*}
 J_\epsilon Pf (\omega_i)=(G_\epsilon Pg)\circ X (\omega_i)= Pg\circ X (\omega_i)+ \mathcal{O}(\epsilon) = Pf (\omega_i)+ \mathcal{O}(\epsilon),
\end{align*} 
as $\epsilon \rightarrow 0$, due to the asymptotic expansion in \eqref{VBDMexpansion}. Together with \eqref{VBDMdiscreteestimate} and \eqref{changeofvariable}, we have
\BEA
\left(\mathbf{G}_{N,\epsilon} \vec{g}\right)_i - Pg(x_i)  &=&   \left(\left(\mathbf{G}_{N,\epsilon} \vec{g}\right)_i -  G_\epsilon g(x_i) \right) + \left(G_\epsilon g(x_i)   -Pg(x_i)\right) \nonumber \\ &=&
\left(\left(\mathbf{G}_{N,\epsilon} \vec{g}\right)_i -  G_\epsilon g(x_i) \right) + \left(J_\epsilon f(\omega_i)   -Pf(\omega_i)\right) \nonumber \\ &=&  \mathcal{O} \left(\epsilon,\frac{q(x_i)^{1/2+d/4}}{N^{1/2}\epsilon^{2+d/4}},\frac{q(x_i)^{d(d/2-1/4)}}{N^{1/2}\epsilon^{1/2+d/4}}\right).\nonumber
\EEA
\end{proof}
 
 For a function $\vec{g}$ whose components are function values at the training data $\{x_1,\ldots, x_N\}$,  the kernel smoothing estimate of $g(x_{out})$ on a new point $x_{out}$ is given by $\mathbf{G}_{N_{out},\epsilon} \vec{g}$, where $\mathbf{G}_{N_{out},\epsilon}$ is a row vector consisting of $\kappa_\epsilon(x_{out},x_i)$ for $i=1,\ldots, N$. 
  The definition of $G_{\epsilon}$ can be extended to $g \in \mathbb{R}^n$, where $n>1$ componentwise. That is, if $g(x)=(g_1(x),\ldots,g_n(x))$, where $g_i (x) \in \mathbb{R}$ for $1 \leq i \leq n$ then $G_{\epsilon}g(x):= (G_{\epsilon}g_1(x),\ldots, G_{\epsilon}g_n(x))$. Then, the componentwise convergence in probability holds due to the preceding proposition. The discrete estimator then becomes a matrix-matrix multiplication, $\mathbf{G}_{N,\epsilon}\mathbf{g}$, where the $ij$th component of the matrix $\mathbf{g}$ is given by $g_{i}(x_j)$.  

The kernel smoothing estimate is conceptually simple and computationally fast to construct.  While naive, we will show in the next section that the kernel smoothing estimate of the conditional expectation performs well when accurate estimation of the eigenbasis of $V$ is not available, especially when the covariate space is high-dimensional.

\section{Predicting the dynamics of observables}
\label{sec:prediction}
In this section, we discuss the problem of predicting observables (e.g.,partial components) of a measure preserving discrete time dynamical system. We start by reviewing the MZ formalism \cite{Zwanzig2001}, which is a classical reduced-order modeling framework often used for this task. The MZ formalism expresses the evolution of the desired reduced order dynamics in terms involving Markovian, non-Markovian, and orthogonal dynamics through the use of an orthogonal projection operator. The total contributions of the Markovian and non-Markovian terms in this decomposition will coincide with the optimal solution to the regression problem with observables at initial and future times as covariate and response data, respectively. For a low-dimensional covariate space, we show that the regression estimator can be accurately constructed using the Nystr\"{o}m method in Section~\ref{sec32}. 
In Section~\ref{sec:estMZ}, we argue that if the hypotheses of delay-embedding theorems are satisfied \cite{TAKENS2010345,MR1137425}, the MZ-equations can be simplified to only a ``Markovian" term, whose representation is precisely the regression function that maps the delay-embedded observable to the observable at a future time. This regression function, which can be estimated by KAF, is nothing but the component of the flow map induced by the lag embedding. In such high-dimensional covariate space regression problems, we numerically demonstrate that the kernel smoothing estimate is a more accurate estimator than the Nystr\"{o}m method.

\subsection{Mori-Zwanzig formalism for reduced order modeling}\label{sec31}

Let $(\Omega, \Phi)$ be a discrete-time deterministic dynamical system, generated by an invertible map $\Phi: \Omega\to\Omega$. Furthermore, we assume that there is a $\Phi$-invariant probability measure $\mu:\mathcal{B}(\Omega)\rightarrow [0,1]$, where $\mathcal{B}(\Omega)$ is the Borel $\sigma$-algebra on $\Omega$. That is, for all $B\in\mathcal{B}(\Omega)$, $\mu(\Phi^{-1}(B))=\mu(B)$. As in Section~\ref{sec:regression}, we assume that $\mu$ is supported on a compact set $M \subseteq \Omega$. For a given $\omega_0\sim \mu$, we let $\omega_i:= \Phi^i(\omega_0)$.

In what follows, we assume that only partial observations $x_i\in \mathbb{R}^n$ of $\omega_i$ are available and they are defined through a measurable  function $X:\Omega \rightarrow \mathcal{X} = \mathbb{R}^n$ such that $x_i:= X(\omega_i)=X\circ\Phi^i(\omega_0)$. Since our goal is to use the observed time series of $\{x_i\}$ to estimate $x_{i+t}\in\mathcal{X}$ for some $t \in \mathbb{N}$, we set the response space equals to the covariate space, $\mathcal{Y}=\mathcal{X}$, and consider the response map $Y = X_t:\Omega \rightarrow \mathcal{X}$, where $X_t = X \circ \Phi^t $; that is, $ X_t(\omega_0) =  X(\omega_{t})$. Note that since $\mathcal{X}=\mathcal{Y} = \mathbb R^n$, we have $H=\{f:\Omega\to \mathcal{X}: \int_\Omega \lVert f^2(\omega) \rVert^2d\mu(\omega)< \infty\}$, $V=\{g:\mathcal{X}\to\mathcal{X}: g\circ X\in H\}$, and  $H_X=\{f\in H: f=g\circ X {\mbox{ for some } g\in V} \}$.

Let us define an orthogonal projection operator $P:H\to S_X \subseteq  H_X\subseteq H$, where $S_X = \mbox{ran} (P)$ is a closed subspace of $H_X$, and let $Q=I-P$ be the orthogonal projection onto the orthogonal space, $S_X^\perp=\mbox{null}(P)$. With these projection operators, we have $H= S_X \oplus S_X^\perp$. Let $U:H\to H$ be the Koopman operator defined as $Uf = f\circ \Phi$, for all $f\in H$. Then by ``the Dyson's formula" \cite{darve2009computing,lin2019data}, the map $U^{i}:H\to H$ given by $U^{i}f=f\circ \Phi^i $ can be written as,

 \begin{equation}
U^{i+1}=\sum_{k=0}^iU^{i-k}PU(QU)^k+(QU)^{i+1}. \label{Dyson}
\end{equation}

Applying \eqref{Dyson} on $X\in H$, we obtain the discrete MZ equation that describes the evolution of $x_i=X(\omega_i)$. In detail, letting $\Xi_{i}:=(QU)^{i}X$ and noting that
\begin{align*}
U^{i+1}X&=X_{i+1}\\
U^{i-k}PU(QU)^kX&=PU(QU)^kX\circ\Phi^{i-k} =P(\Xi_k\circ \Phi)(X\circ\Phi^{i-k})\\
(QU)^{i+1}X&=\Xi_{i+1},
\end{align*}
along with the fact that $PQ=0$, yields
\begin{eqnarray}
X_{i+1}(\omega_0)= (PU X \circ \Phi^i)(\omega_0)+\sum_{k=1}^iP( \Xi_k\circ \Phi)\circ(X\circ \Phi^{i-k})(\omega_0)+\Xi_{i+1}(\omega_0). \label{MZ}
\end{eqnarray}

Since $PUX \in S_X$, there exists an $M_0 \in V$ such that $PUX = M_0\circ X$, by the definition of $S_X$. Similarly, since $P (\Xi_k\circ \Phi)\circ X \in S_X$, there exists $M_k\in V$ such that $P (\Xi_k\circ \Phi)\circ X = M_k \circ X$. Therefore, Eq.~\eqref{MZ} can be written in terms of the observable values $x_i$ as
\BEA
x_{i+1}=M_0(x_i)+\sum_{k=1}^iM_k(x_{i-k})+\Xi_{i+1}(\omega_0). \label{MZ2}
\EEA

Note that \eqref{MZ2} decomposes $x_{i+1}$ into the Markovian term $M_0$,  the memory terms $M_i$ and a term $\Xi_i$ that is orthogonal to $S_X$.

\subsection{Approximation of the projected Mori-Zwanzig equation}\label{sec32}

If we consider the specific choice $S_X = H_X$, and thus the projection operator $P:=\mathbb{E}[\cdot| X]$, we obtain the projected MZ-equation,
\begin{eqnarray}
\mathbb{E}[X_{i+1}|x_0] = M_0(x_i)+\sum_{k=1}^iM_k(x_{i-k}), \label{PE}
\end{eqnarray} 
since $M_k\circ X\in H_X$, and the orthogonal term $P\Xi_{i+1}=0$ since $\Xi_{i+1}\in H_X^\perp$.  In this case, notice that $\mathbb{E}[X_{i+1}\mid \cdot]$ is precisely the minimizer in \eqref{optimal} with $X_{i+1}$ in place of $Y$. The main takeaway here is that the solutions of the projected MZ-equation in \eqref{PE} is the regression function, $\mathbb{E}[X_{i+1}|\cdot]$, of the dynamical map $X_0 \mapsto X_{i+1}$. The importance of this observation is that one can approximate $\mathbb{E}[X_{i+1}|X_0]$ from the historical data $\{x_i\}$. In the next two examples, we will numerically verify this claim with the two nonparametric estimators discussed in Section~\ref{sec:regression}, the Nystr\"om method and kernel smoothing.
  
{\bf Hamiltonian system:} First, consider the $16$-dimensional dynamical system given by the Hamiltonian 
\BEA
H(\omega)=\frac{1}{2}\bigg(\sum_{i=1}^{16}\omega_{(i)}^2+\sum_{i=1}^{7}\omega_{(2i-1)}^2\omega_{(2i+1)}^2\bigg), \label{HH}
\EEA
where  $(\omega_{(2i-1)},\omega_{(2i)})$ for $i=1,\ldots,8$ are the canonical conjugate variables and $\omega=(\omega_{(1)},\ldots,\omega_{(16)})\in\mathbb{R}^{16}$. Thus the full system is derived through the relations  
\begin{eqnarray}
\frac{d{\omega}_{(2i-1)}}{dt}=\frac{\partial H(\omega)}{\partial \omega_{(2i)}},\quad\quad \frac{d{\omega}_{(2i)}}{dt}=\frac{\partial H(\omega)}{\partial \omega_{(2i-1)}}. \label{Hamiltonian}
\end{eqnarray} 
In addition to the subscript-$(i)$ used to denote the $i$th component of $\omega\in\mathbb{R}^{16}$, we will use the notation $\omega_j$ to denote the $j$th sample of $\Omega$ with sampling measure $\frac{d\mu}{d\omega} \propto e^{-H(\omega)}$, where $\omega_j = (\omega_{j,(1)},\ldots,\omega_{j,(16)})$. We should point out that this example is a high-dimensional version of the main example in~\cite{chk:02}.

Suppose we are interested in the conditional expectation $\mathbb{E}[\omega_{(1)}(t),\omega_{(2)}(t)\mid \omega_{(1)}(0),\omega_{(2)}(0)]$, where the expectation is drawn from the canonical invariant density $\mu$ corresponding to $H$ with  fixed $\omega_{(1)}(0)$ and $\omega_{(2)}(0)$. We can view the problem of estimating the conditional density in the regression framework as follows. Let us define the covariate map $X:\Omega\rightarrow \mathcal{X}$ by $X(\omega(0))=(\omega_{(1)}(0),\omega_{(2)}(0)):=x_0 \in \mathcal{X}$ and the response variable $X_t(\omega)=X(\omega(i\Delta t))=(\omega_{(1)}(i\Delta t),\omega_{(2)}(i\Delta t))$, where $\Delta t>0$ is a fixed time step.

In this example, we will consider estimates based on the Nystr\"om method with $L=100$ and the kernel smoothing estimator. For this application, let $\vec{x}_0 = (x_{0,1},\ldots,x_{0,N})$ and $\vec{x}_0^{out}= (x_{0,1}^{out},\ldots,x_{0,N_{out}}^{out})$ be two vectors that will be used for training and verification, respectively. Each component of these vectors is an i.i.d sample of $X_0$, that is, $x_{0,j}=X_0(\omega_j)$, where $\omega_j$ is drawn independently from $\mu$. Let $\vec{x}_t:=(x_{t,1},\ldots,x_{t,N})$ be a vector of the training time series with components given by $x_{t,j}=U^t\circ X_0(\omega_{j})$.

In the following numerical experiments, the time series $\{x_{t,1},\ldots,x_{t,N}\}_{t=0,1,\ldots}$ was observed at the sampling interval $\Delta t=.1$ time units, and the initial conditions $\{x_{0,1},\ldots,x_{0,N}\}$ are samples of the invariant density $\nu=X_*\mu$. We verify the quality of the estimators on $N_{out} = 1000$ out-of-sample initial conditions, also sampled from $\nu$. To verify the performance of the two estimators, we compare them to the empirical conditional expectation obtained from a Monte-Carlo simulation. This calculation requires samples of the conditional distribution of $(\omega_{(3)}(0), \ldots,\omega_{(16)}(0))$ given each out-of-sample initial condition, $(\omega_{(1)}(0),\omega_{(2)}(0))=x_{0,j}$, where $j=1,\ldots, N_{out}$. Numerically, we obtain these samples, denoted by $(\omega_{(3),j}^{(k)}(0), \ldots,\omega_{(16),j}^{(k)}(0))$, using the Hamiltonian Monte Carlo method \cite{betancourt2017conceptual} on the reduced Hamiltonian in \eqref{HH} with fixed $(\omega_{(1)}(0), \omega_{(2)}(0))=x_{0,j}$. Concatenating these samples and the fixed $x_{0,j}$, we define $\omega_j^{(k)}=(x_{0,j},\omega_{(3),j}^{(k)}(0), \ldots,\omega_{(16),j}^{(k)}(0))$, for $j =1,\ldots, N_{out}, k=1,\ldots,N_{MC}$. In the numerical result below, we use $N_{MC}=20,000$ samples for each initial condition $x_{0,j}$. Given these samples, the Monte-Carlo approximation of $\mathbb{E}[X_t\mid \cdot ]$ is given by 
\BEA 
\mathbb{E}[X_t\mid x_{0,j}]\approx \frac{1}{N_{MC}}\sum_{k=1}^{N_{MC}}U^t\circ X(\omega_j^{(k)})=\frac{1}{N}\sum_{k=1}^{N_{MC}}X\circ \Phi^t(\omega_j^{(k)}), \label{MCMC}
\EEA
where each realization $\Phi^t(\omega_j^{(k)})$ is the solution of the full dynamics with solution map denoted by $\Phi$. Note that the solution map of \eqref{HH} is the result of a temporal discretization of the Hamiltonian dynamics in \eqref{Hamiltonian}; in our numerics, we use the Runge-Kutta-4 (RK4) method. 

In Figure~\ref{Hamiltonian16D}, we show a comparison of the Nystr\"om method and the kernel smoothing estimate of the conditional expectation, constructed using $N=20,000$ training samples and the empirical Monte-Carlo estimate in \eqref{MCMC} for a particular out-of-sample data (which is considered as the truth). Notice that the Nystr\"om method is significantly more accurate than the kernel smoothing estimate. In Figure~\ref{Hamiltonian16D}(b), we also show the Root-Mean-Square-Errors (RMSEs) between the two estimators and the empirical Monte-Carlo estimate of the conditional expectation, averaged over $N_{out}=1000$ out-of-sample initial conditions for training data sizes, $N=10,000$ and $20,000$. Besides the clear advantage of the Nystr\"om method over kernel smoothing, notice that both estimators are improved as the size of training data, $N$, increases.  

Although the Nystr\"om method performs better than the kernel smoothing estimate, the former is computationally more expensive expensive. For both methods we construct $\mathbf{G}_{N,\epsilon}$ using the steps in \eqref{dmnormalization}. To alleviate memory and computation costs associated with full $N\times N$ kernel matrices for $N\gg1$, in practice a $k$-nearest neighbor algorithm is employed so the resulting matrix $\mathbf{G}_{N,\epsilon}$ is sparse with approximately $k$ nonzero entries on each row. To predict on a new data point using the kernel smoothing method, one only needs to extend $\mathbf{G}_{N,\epsilon}$ to the new point and multiply the extended row vector by the column of the response training data. On the other hand, the Nystr\"om method requires the eigenvectors of the matrix $\mathbf{G}_{N,\epsilon}$ and then employs the Nystr\"{o}m extension method to approximate the eigenfunctions evaluated on the new out of sample data points. Thus, after constructing $\mathbf{G}_{N,\epsilon}$, the computational cost of the kernel smoothing method is $\mathcal{O}(k)$ where $k$ is the number of nearest neighbors employed, while the cost for the Nystr\"{o}m method is $\mathcal{O}(kL)+\mathcal{O}(\mathrm{E.D.})$, where $L$ is the number of eigenfunctions used and $\mathcal{O}(\mathrm{E.D.})$ is the cost of acquiring the $L$ eigenvectors.

\begin{figure}
{\ \centering
\begin{tabular}{cc}
(a) & (b)  \\
\includegraphics[width=.49\linewidth]{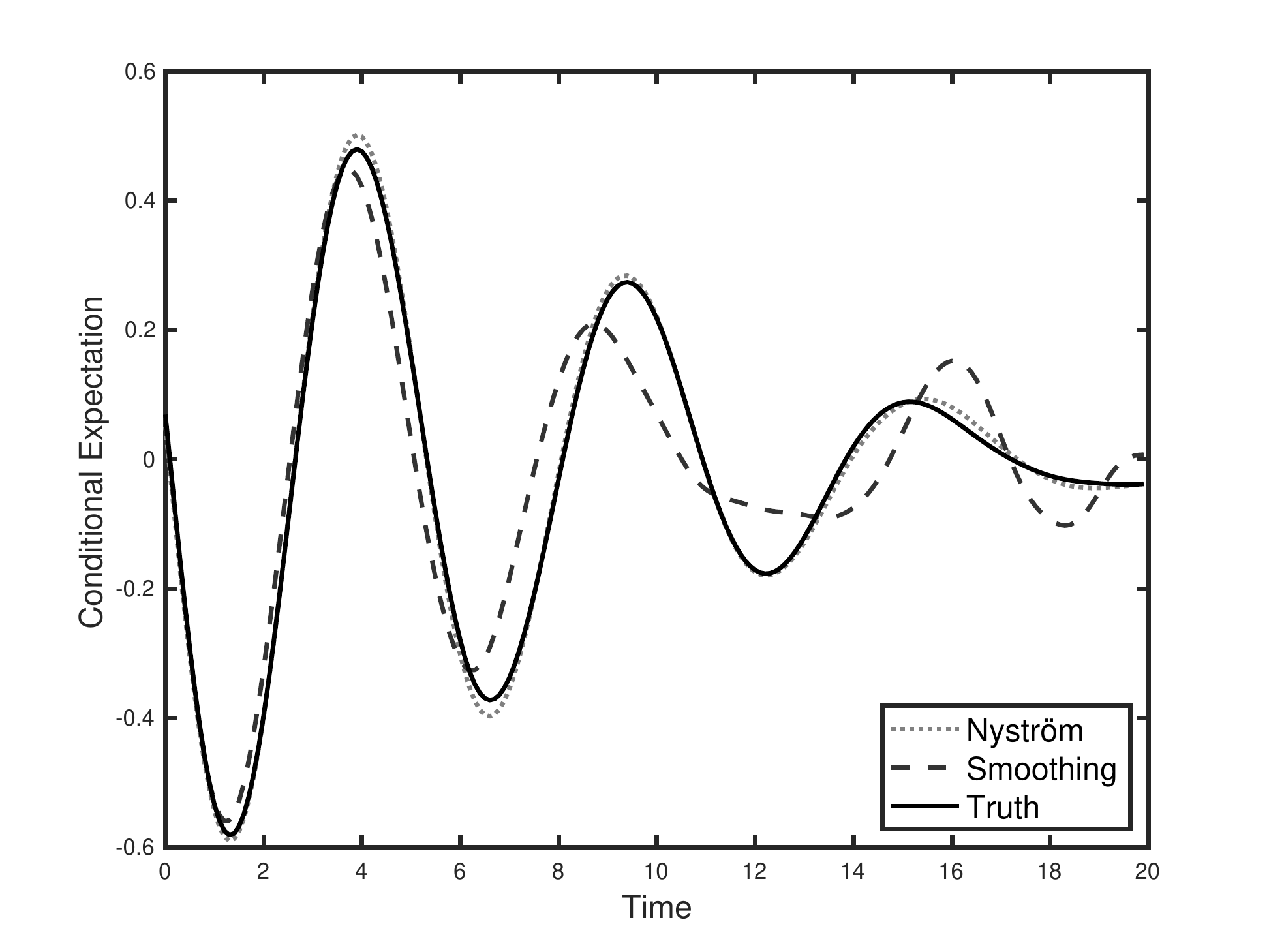}
&
\includegraphics[width=.49\linewidth]{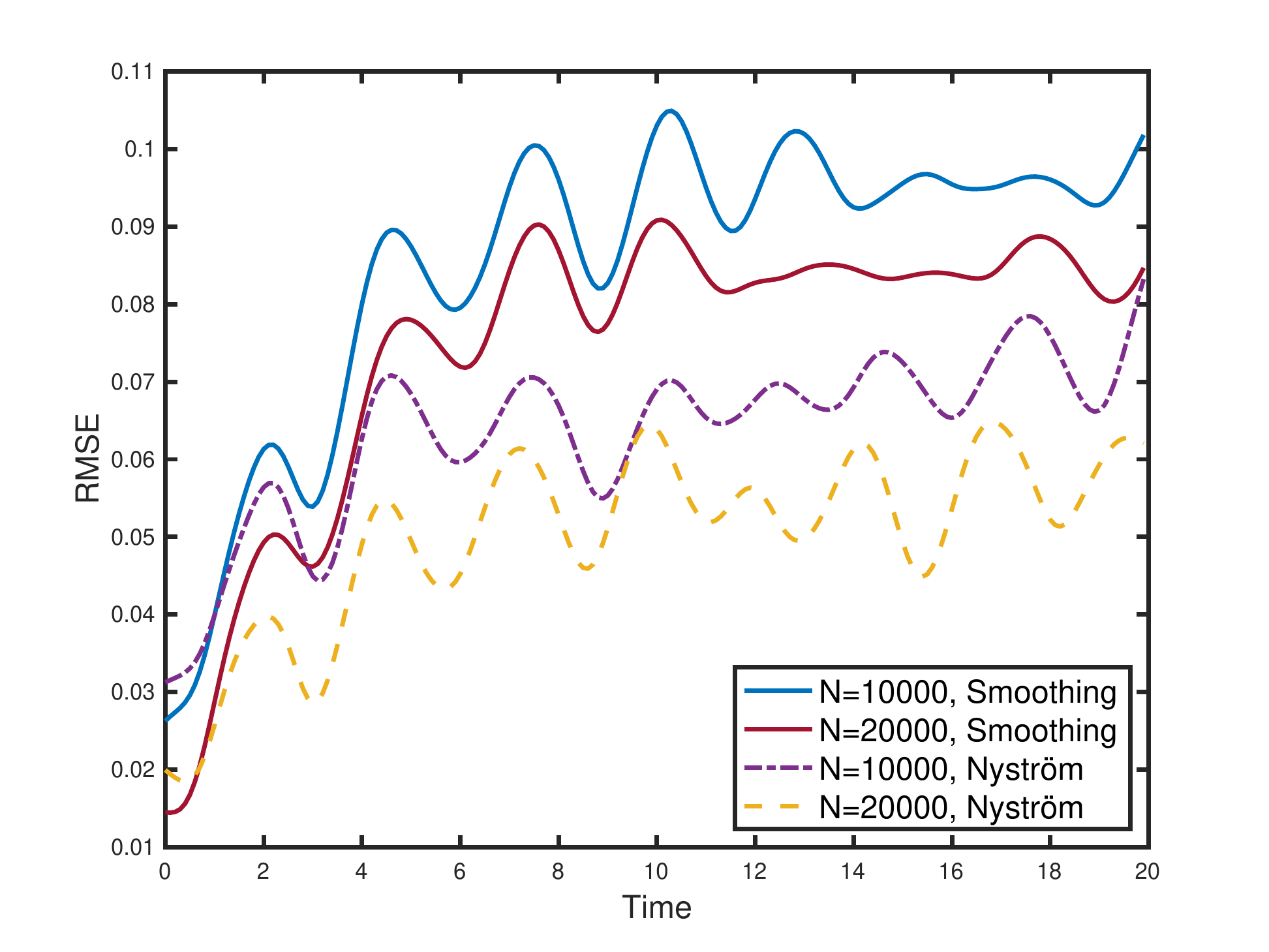}
\end{tabular}%
}
\caption{Hamiltonian Example: (a) Comparison of the kernel smoothing estimate (Smoothing), the Nystr\"om method estimate, and the MC empirical estimate (which is considered as the truth) of the conditional expectation of the first component of a particular out-of-sample trajectory, trained using $N=20,000$ samples. (b) The RMSEs (based on $N_{out}=1000$ samples) of both estimators as functions of lead time forecast, constructed using $N=10,000$ and $N=20,000$ data points.}
\label{Hamiltonian16D}
\end{figure}

{\bf The Lorenz-96 model:} Next, we consider the Lorenz-96 model \cite{Lorenz:96} given by 
\begin{eqnarray}
\frac{d\omega_{(i)}}{dt}=(\omega_{(i+1)}-\omega_{(i-2)})\omega_{(i-1)}-\omega_{(i)}+F \label{L96model}
\end{eqnarray}

for $i=1,\ldots,5$, forcing parameter $F=8$ and, with periodic boundary condition, $\omega_{(-1)}=\omega_{(4)}, \omega_{(0)}=\omega_{(5)},$ and $\omega_{(6)}=\omega_{(1)}$. In this regime, the dynamics is chaotic with attractor dimension 2.9 and two positive Lyapunov exponents as reported in \cite{gottwald2013mechanism}. We estimate the conditional expectation $\mathbb{E}[X_t\mid X_0]$, where the covariate function is $X(\omega)=\omega_{(1)}(0)$, the response function is $X_t(\omega)=X(\Phi^t(\omega))=\omega_{(1)}(t)$, and the initial conditions are drawn from the standard Gaussian distribution. Note that this distribution is not invariant under the dynamics of the system. Here $\Phi$ is given by the RK4 discretization of \eqref{L96model} with time step $1/64$. 

Numerically, we generate $N=20,000$ and $N_{out}=1000$ initial conditions for training and verification, respectively, from the standard five-dimensional multivariate Gaussian and integrate the training data forward $2.5$ time units to generate training time series observations. Subsequently, we used only the first component, $\omega_{(1)}$, of the initial conditions and the training time series to construct, both, the Nystr\"om and kernel smoothing estimates of the conditional expectation. For the Nystr\"om method, we use $L=300$ eigenfunctions. Both estimators are compared to an empirical estimator which is obtained by averaging \eqref{MCMC} over $\omega_j^{(k)}=(x_{0,j},\omega_{(2),j}^{(k)}(0), \ldots,\omega_{(5),j}^{(k)}(0))$, for $j =1,\ldots, N_{out}, k=1,\ldots,N_{MC}=20,000$ samples of initial conditions. Here, the first component of each initial condition, $x_{0,j}=X(\omega_j^{(k)})$, is one of the $N_{out}=1000$ verification samples and the other components are drawn from the four-dimensional standard Gaussian. In Figure~\ref{L96MVN}(a), we show the evolution of one of the $1000$ verification samples. Comparing the three estimates, notice the closer agreement between the Nystr\"om and the empirical estimates. In Figure~\ref{L96MVN}(b) one can see that the RMSE (based on averaging over $N_{out}=1000$ out-of-sample points) of the Nystr\"om-based estimate is more accurate than the kernel smoothing estimate. This result is consistent with the previous example.

\begin{figure}
{\ \centering
\begin{tabular}{cc}
(a) & (b)  \\
\includegraphics[width=.49\linewidth]{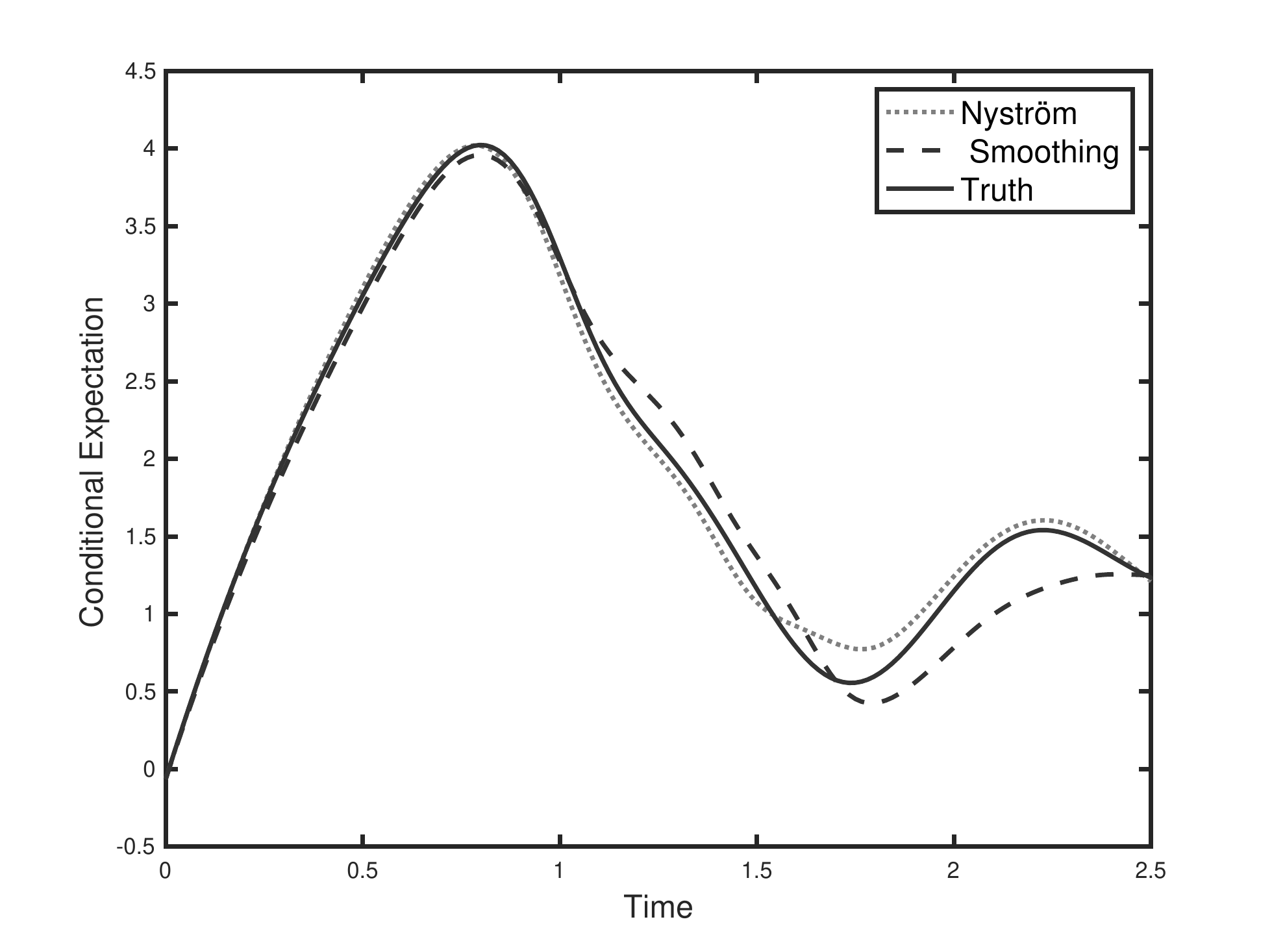}
&
\includegraphics[width=.49\linewidth]{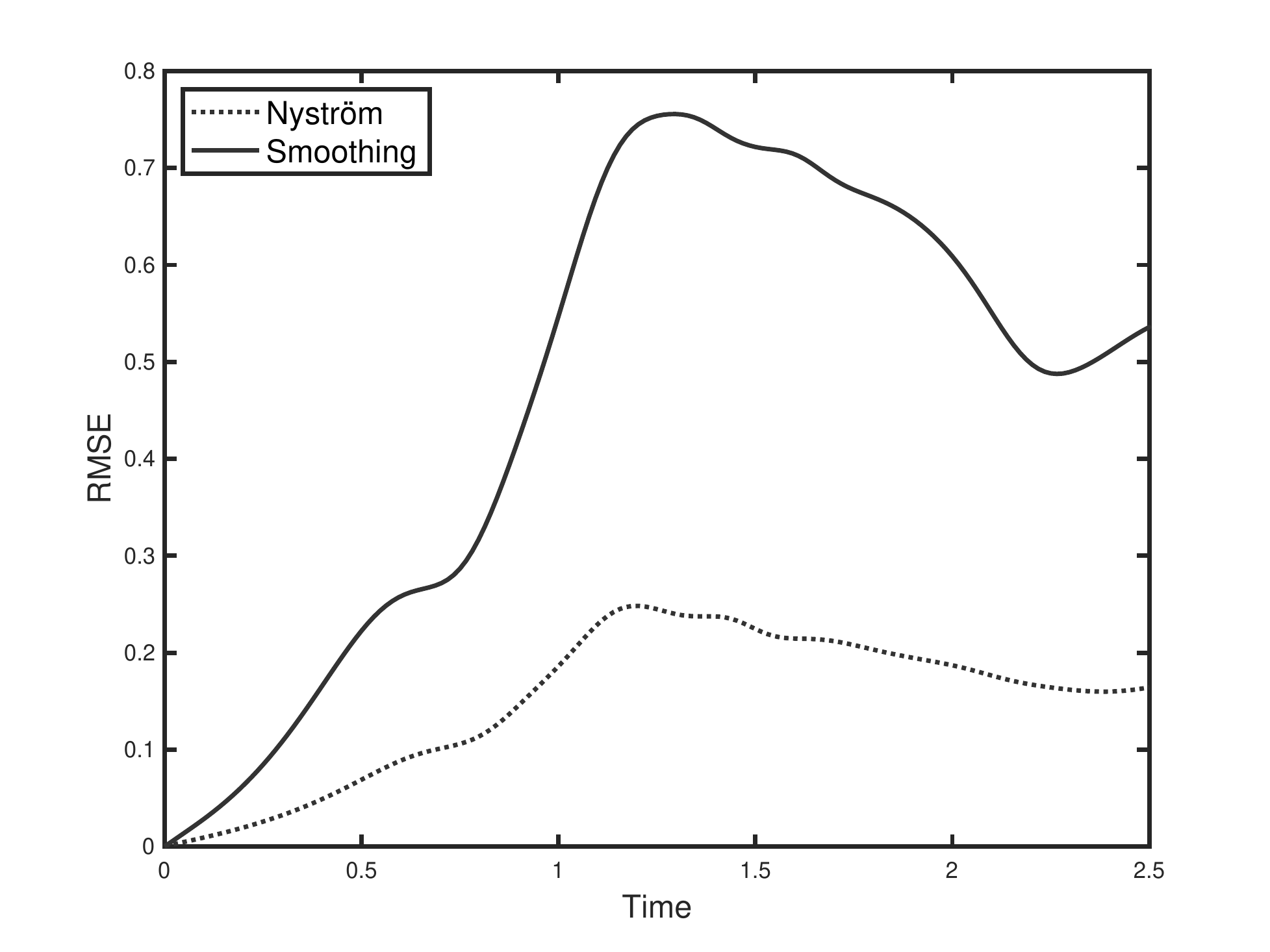}
\end{tabular}%
}
\caption{The Lorenz-96 example: (a) Comparison of the kernel smoothing estimate (Smoothing) and the Nystr\"om method of the conditional expectation of the first component of a particular out-of-sample trajectory, trained using $N=20,000$ training data. (b) The RMSEs (based on $N_{out}=2000$) of both estimators as functions of lead time forecast, constructed using $N=20,000$ data points.}
\label{L96MVN}
\end{figure}

Next, we will show that the full MZ equation \eqref{MZ}  can be constructed by an optimal least squares estimator of a regression framework with appropriate choice of covariate space. 

    \subsection{Mori-Zwanzig projection and delay-coordinate maps}\label{sec:estMZ}

    In the preceding subsection, we considered estimating the conditional expectation $\mathbb{E}[X_{t}\mid X_0]$  and showed that it can be numerically approximated using the time series data. In this subsection, we are interested in predicting the realization of $x_t$ in \eqref{MZ2}. While the MZ representation suggests that the  solution depends on the entire historical data, for practical computation, finite-memory models to collectively represent these terms as a finitely supported function is desirable. Since the memory terms depend on the orthogonal dynamics (see Eq.~\eqref{MZ}), such an approximation can be achieved, e.g, by delta function approximation \cite{hijon2010mori}, Krylov subspace approximation \cite{LiXian2014}, rational approximation \cite{hl:15}, or a series representation of the orthogonal dynamics \cite{LiStinis2019,zhu2018faber}. Note that while a finite-dimensional (matrix) representation is the computational object of interest, a series representation may not converge since it involves expansion of semigroups generated by unbounded operators.

On the other hand, we should point out that depending on the choice of the projection operator $P$, the explicit representation of the terms in the MZ equation, $(M_j)_{j=0}^{t-1}$ as well as the orthogonal dynamics $\Xi_t$, may or may not be easily translated into an efficient algorithm that yields a consistent approximation. As we showed in the preceding subsection, choosing $P=\mathbb{E}[\cdot|X]$ as an estimator will not yield an accurate approximation to $x_t$ since this estimator truncates the orthogonal dynamics. Other common choice of projection operators can be found in \cite{chk:02,Zwanzig2001}. For example, while the popular Mori projection, defined as $P = \langle X ,X\rangle_{H}^{-1}\langle \cdot,X\rangle_{H}X$, yields a linear model for $(M_j)_{j=0}^{t-1}$, the representation of the orthogonal dynamics in such a basis expansion may not be computationally tractable \cite{chuli:2019}. 

Recently, it was shown in \cite{lin2019data} that by choosing $P$ to be the Wiener projection, one can simplify the MZ equation so that only the Markovian term $M_0$ and orthogonal terms $\Xi_t$ remain, where $M_0$ is now a function that takes a delay coordinate of the observable. Building on this result, our intuition is to construct a non-decreasing sequence of projection operators $\{P_m: m\in \mathbb{N}\}$ which allows one to access the entire function space $H$ with a finite $m$ and a simple representation of the MZ equation. In what follows, we argue that delay-embedding theorem \cite{TAKENS2010345} provides a natural candidate for achieving this goal.    

To that end, we define the delay coordinate map $\mathbb{X}_m:\Omega\rightarrow \mathcal{X}^m$  by  $\mathbb{X}_m(\omega)=(X_{{-m+1}}(\omega),\ldots,X_{-1}(\omega),X_{0}(\omega))$, $X_i=U^i\circ X$, which we will consider as the covariate function. Simultaneously, we consider the response function $X_{t}:\Omega\to \mathcal{X}$, where $\mathcal{X}$ is the response space as in the preceding sections. Note that the optimal estimator for the map $\mathbb{X}_m\mapsto X_{t}$ is given by the conditional expectation $P_mX_t:=\mathbb{E}[X_{t}\mid \mathbb{X}_m]$. Under mild assumptions on the covariate $X$, the dynamical flow $ \Phi $, and the sampling interval $\Delta t$, the theory of delay-coordinate maps \cite{TAKENS2010345,MR1137425} states that $\mathbb{X}_m$ is a homeomorphism between the support, $M$, of the invariant measure and $\mathcal{X}^m$ for sufficiently large $m$. Consequently, the Borel sigma algebra on $M$ is identical to the sigma algebra generated by $\mathbb{X}_m$, $\sigma(\mathbb{X}_m)$. Thus, $X_t$ is measurable with respect to the sigma algebra generated by $\mathbb{X}_m$, which means that $P_m X_t:=\mathbb{E}[X_t\mid \mathbb{X}_m]=X_t$ is the identity map for sufficiently large $m$. 

Let $m$ be such that the embedding result stated above holds. Then, letting $P=P_m$ in \eqref{MZ}, the memory and the orthogonal terms vanish since they involve $Q_mUX = (I-P_m)UX = 0$. While $P_m$ is an identity operator, the MZ equation reduces to a contribution of a Markovian term,
\BEA
x_{i+1} = (P_mUX\circ\Phi^i)(\omega_0) = \mathbb{E}[UX | \mathbb{X}_m(\omega_i)] = M_0(x_{i-m},\ldots,x_i),\nonumber
\EEA
for some $M_0 \in V_m:=\{f:\mathcal{X}^m \to \mathcal{X}: f\circ\mathbb{X}_m\in H \}$. If we define the flow map $T$ on $\mathcal{X}^m$ induced by $\Phi$ as $T\circ\mathbb{X}_m(\omega_i):= \mathbb{X}_m\circ \Phi(\omega_i)$, then $M_0$ is the $m$th component of the flow map $T$, which is also the regression function of the supervised learning task $\mathbb{X}_m\mapsto X_t$.
 
In light of this connection, we will employ the nonparametric estimators discussed in Section~\ref{sec:regression} to approximate the regression function $M_0$ and numerically show that true trajectory of the observables can be recovered with adequate accuracy for sufficiently large $m$.

{\bf Hamiltonian system:}  As an example, consider again the Hamiltonian system in \eqref{HH}-\eqref{Hamiltonian}. Here, we are interested in approximating $\mathbb{E}[\omega_{(1)}(t)\mid \omega_{(1)}(-m+1),\ldots,\omega_{(1)}(-1),\omega_{(1)}(0)]$ for $t\in\mathbb{Z_+}$. Letting the response function be $X_t(\omega)=\omega_{(1)}(t)=x_t$ and the covariate function be $\mathbb{X}_m := (X_{-m+1},X_{-1},X_0)$, we can rewrite the conditional expectation of interest as, $\mathbb{E}[X_t|\mathbb{X}_m(\omega)]$. As before, we approximate this conditional expectation using the Nystr\"om method with $L=300$ eigenfunctions and the kernel smoothing estimator. The training data was generated by evolving $N=20,000$ initial conditions $\{\omega_0^{(k)}\}_{k=1,\ldots,N}$, drawn from the invariant measure $\mu$, for $m$ units in time, using RK4 with the same specification as in the previous example.  Figure~\ref{Ham16Dpart}(a) shows a particular out-of-sample trajectory along with the kernel smoothing estimates of the trajectory for various choices of $m$. Notice that as $m$ increases, the kernel smoothing estimator $\mathbb{E}[X_t|\mathbb{X}_m]$ approaches the true trajectory.
 In Figure~\ref{Ham16Dpart}(b), we show the RMSE, averaged over $N_{out}=10000$ out-of-sample verification points. Notice that the RMSEs are smaller as $m$ increases except at initial time. The worse performance at initial time is not so surprising since the kernel smoothing is not an interpolation method, and thus won't be consistent with the given initial conditions.
In panel (c), we show the quality of the prediction for $m=48$  for a particular trajectory. Notice that while the trajectory is well estimated by both methods up to about 6 time units, the kernel smoothing method is more accurate compared to the Nystr\"om method. The improved prediction of the kernel smoothing method compared to the Nystr\"om method at longer times is consistent for different length of memory, $m$, as shown by the RMSE metric in panel (d), computed over $N_{out}=10000$ out-of-sample verification points.

\begin{figure}
{\ \centering
\begin{tabular}{cc}
(a) & (b)  \\
\includegraphics[width=.49\linewidth]{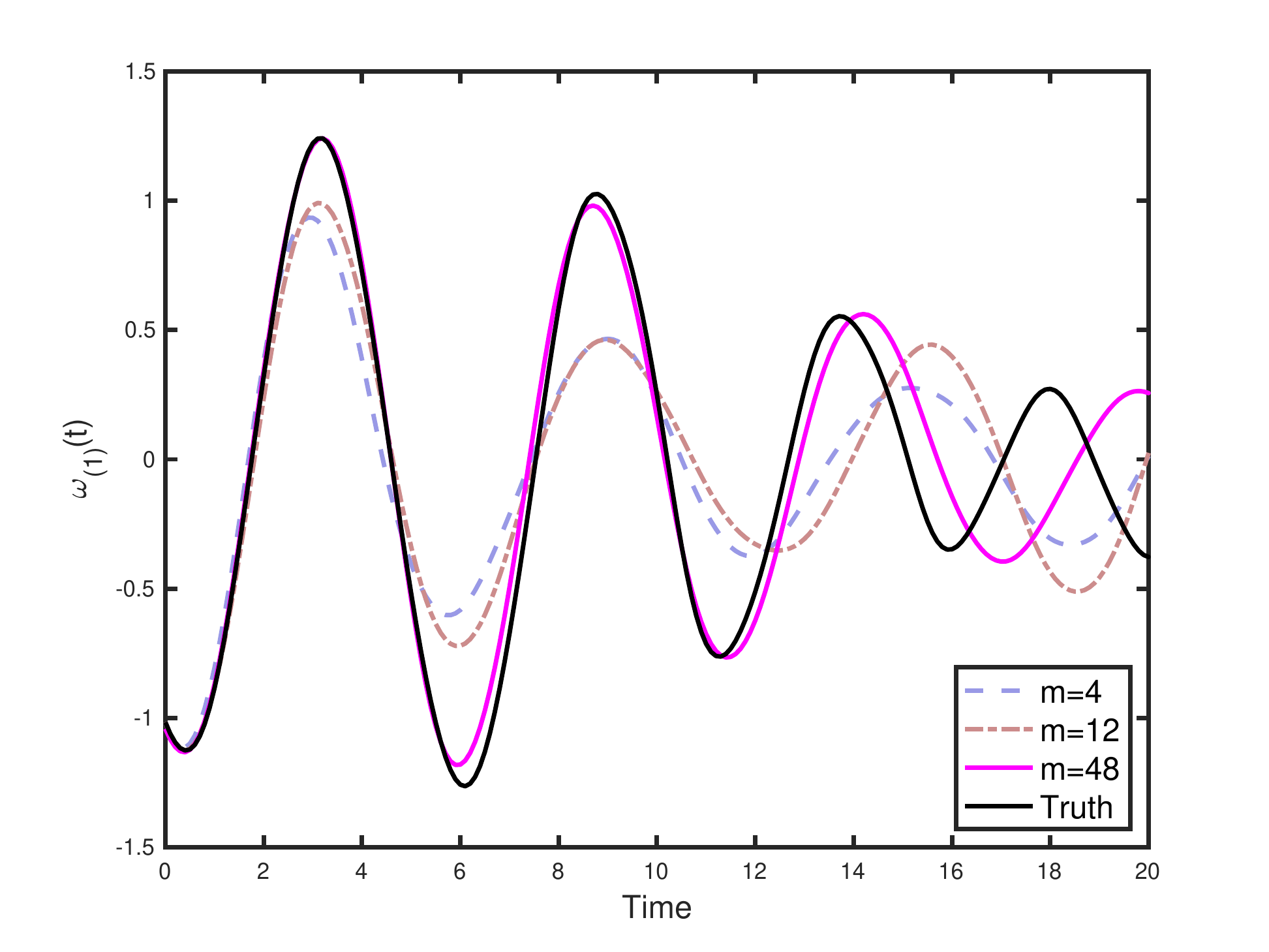}
&
\includegraphics[width=.49\linewidth]{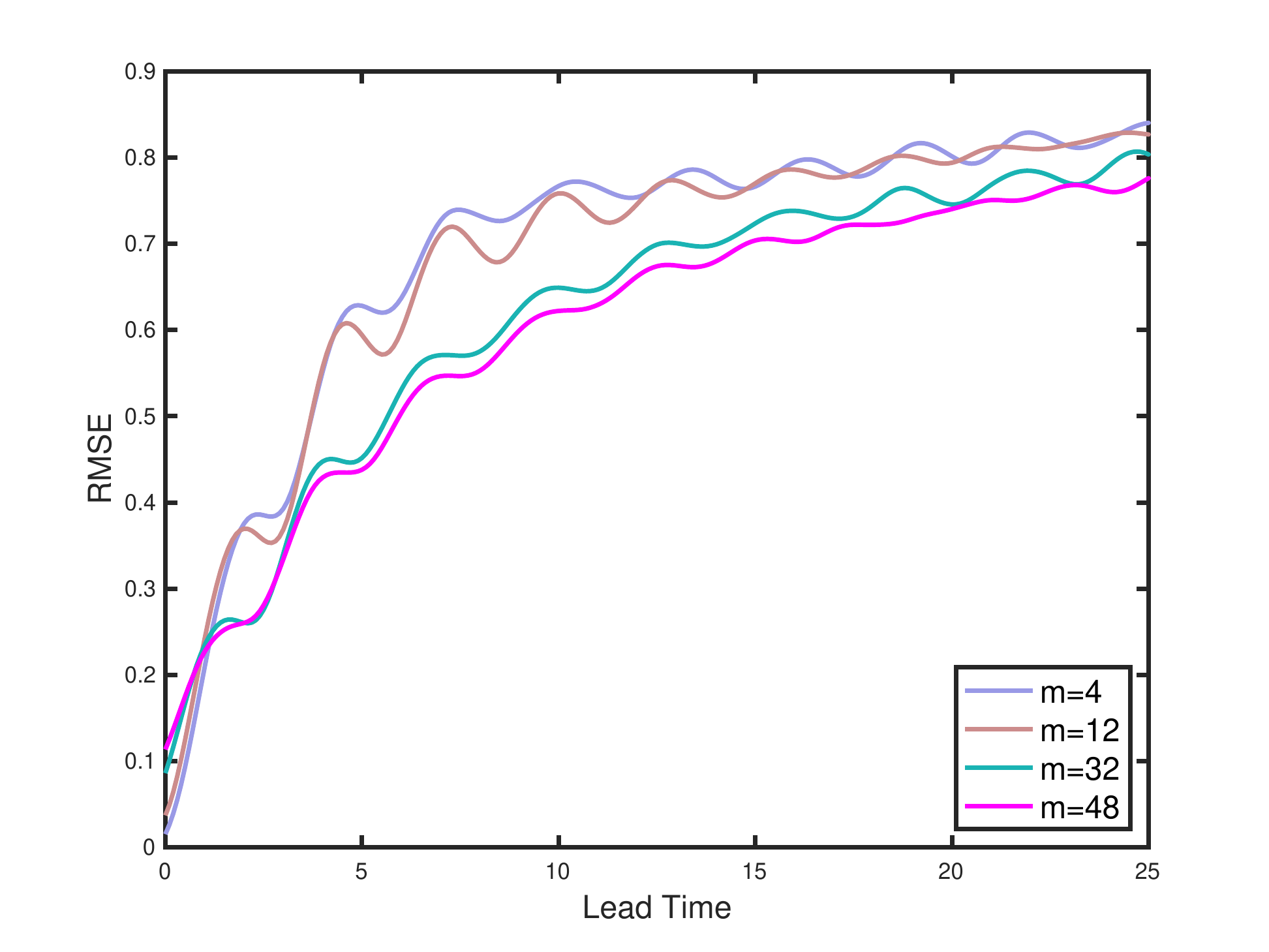}
\\
(c) & (d)  \\
\includegraphics[width=.49\linewidth]{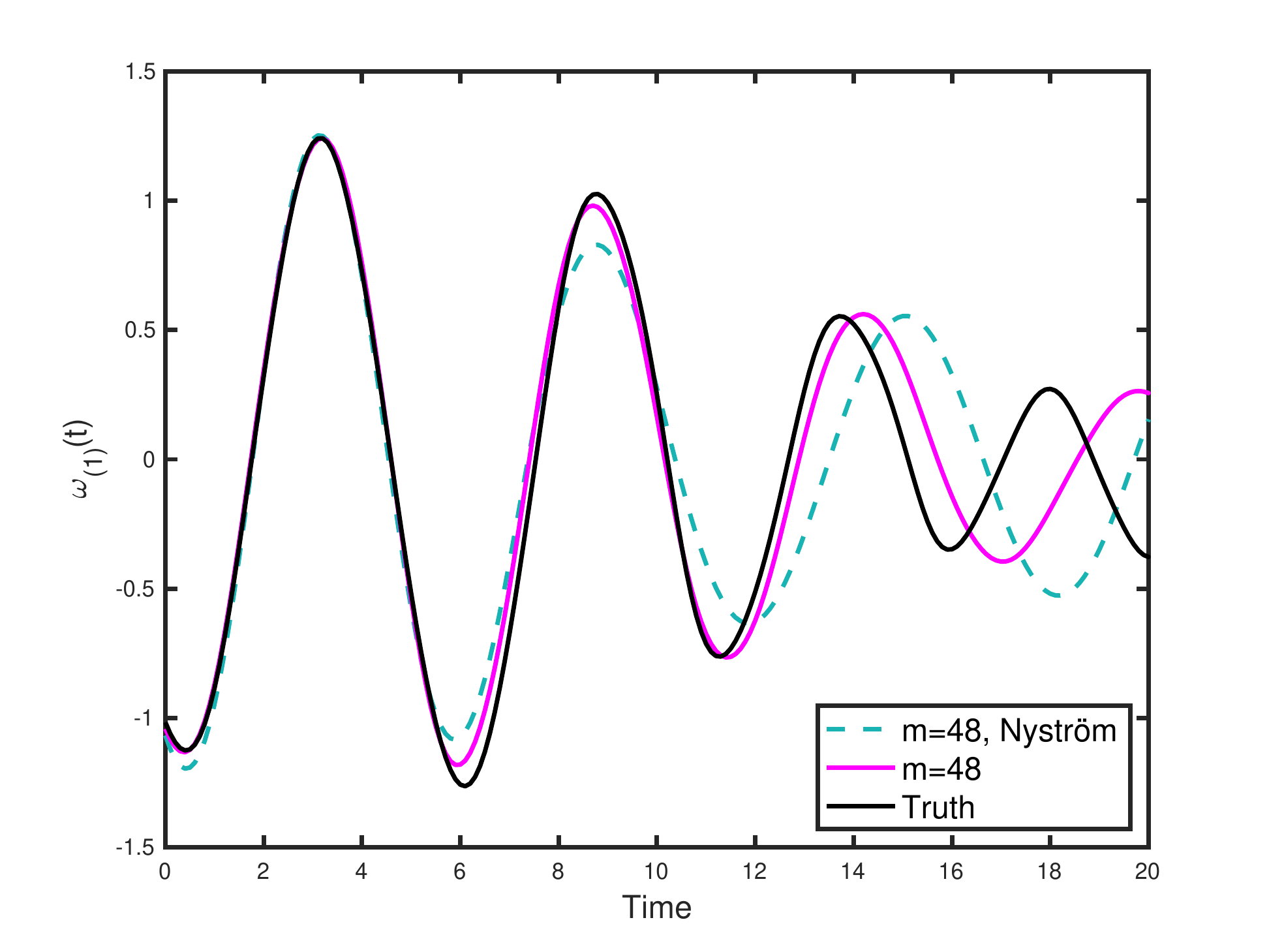}
&
\includegraphics[width=.49\linewidth]{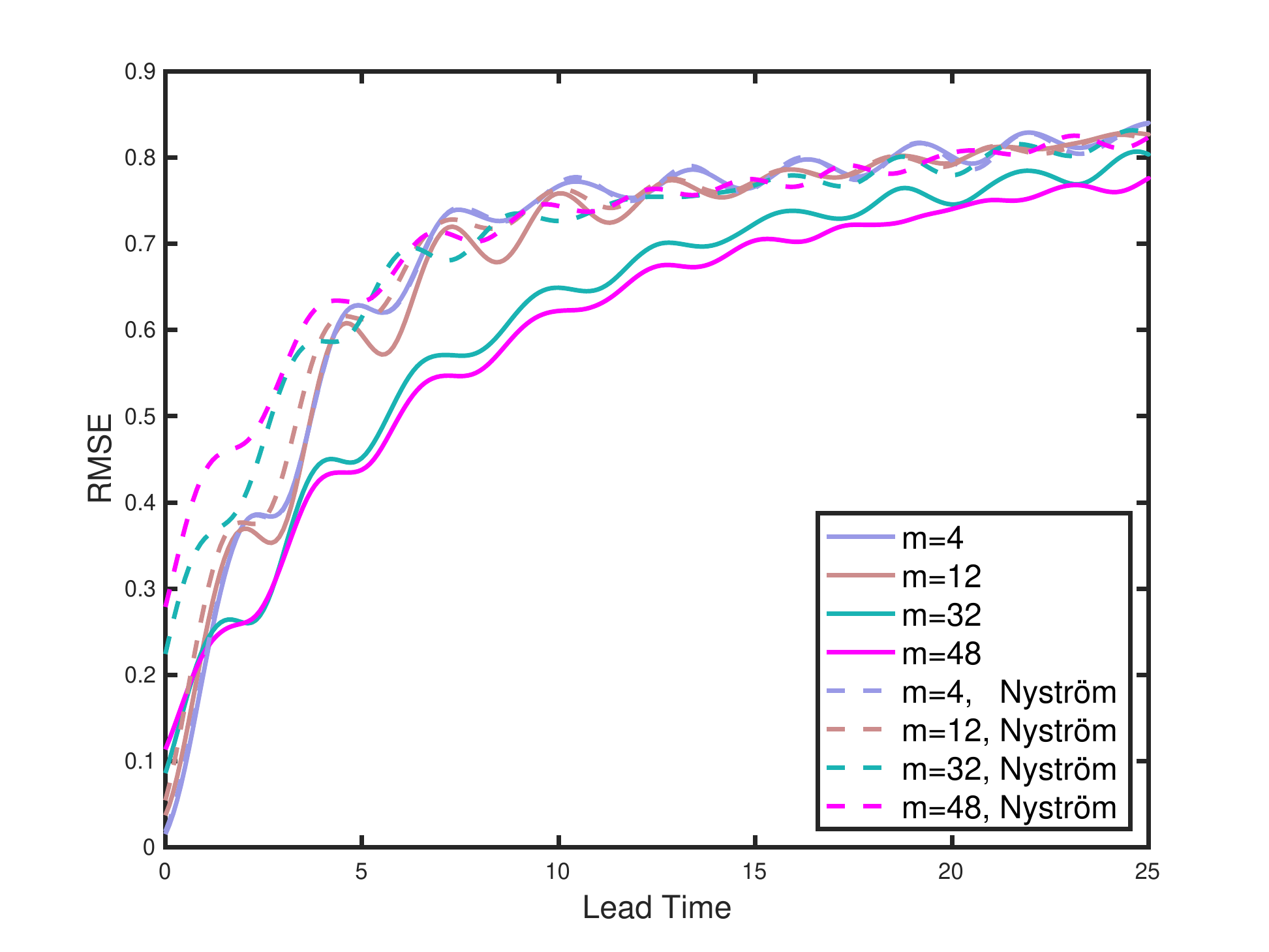}
\end{tabular}%
}
\caption{Hamiltonian Example: (a) The trajectory of the first component, $\omega_{(1)}(t)$ for a particular out-of-sample initial condition along with the kernel smoothing estimates using $m=4, 12,$ and $48$ past observations. (b) The RMSE between the true trajectory and the kernel smoothing estimates of the conditional expectation, calculated over $10,000$ out-of-sample points. (c) A comparison of the kernel smoothing estimate and the Nystr\"om estimate of the trajectory using $m=48$ past data points. (d) The RMSEs of the Nystr\"om and the kernel smoothing estimates of the trajectory for $m=4, 12, 32$, and $48$. Note that the RMSE plots show the RMSE for the lead time and omits the respective training windows for each $m$. }

\label{Ham16Dpart}
\end{figure}

{\bf The Lorenz-96 model:} In this example, we consider predicting the first component $\omega_{(1)}(t)$ of the five-dimensional Lorenz-96 model given by \eqref{L96model}, again with $F=8$.  As in the previous example, we will compare the Nystr\"om method with $L=300$ eigenfunctions and the kernel smoothing method in approximating $\mathbb{E}[X_t|\mathbb{X}_m(\omega)]$, where $X_{t}(\omega)=\omega_{(1)}(t)$. In this example, the delay-embedded data, $\mathbb{X}_m(\omega)$, are sampled from the invariant distribution of the system by running initial conditions sufficiently forward in time. In particular, we take $N=20,000$ samples from the invariant distribution and construct the conditional expectation using observations of the first component of the samples. Here, the time series $U^t\circ X(\omega_1)$ used for constructing the estimator were observed at time steps of $\Delta t=1/64$. In Figure~\ref{L96invPredict}(a), we show a prediction of a particular out-of-sample realization of $\omega_{(1)}(t)$. Notice, as with the previous example, that the quality of the kernel smoothing estimator increases with $m$, except at the initial time as seen in Figure~\ref{L96invPredict}(b). The RMSE in ~\ref{L96invPredict}(b)  was calculated over $N_{out}=10000$ out-of-sample initial conditions, also sampled from the invariant measure. As one can see, the kernel smoothing estimator is consistently more accurate than the Nystr\"om method for similar $m$.

In principle, we should point out that the Nystr\"om method can be improved with a larger number of eigenfunctions $L$. However, there is a practical issue in realizing improved accuracies. In our numerical tests, we do not find any meaningful improvement using any larger $L$ compared to the present results with $L=300$. We suspect that as the covariate space dimension increases (here, controlled by the number of delays $m$),  the Nystr\"om method requires an increasingly higher number of eigenfunctions $L$ to reconstruct the response at a given level of accuracy, and for the available number of training samples $N$, these eigenfunctions cannot be accurately estimated. Thus, unless a mechanism is in place to ensure that the response is well approximated by the leading kernel eigenfunctions, or the  eigenfunctions corresponding to large $L$ can be robustly estimated with modest amounts of data (both of which are highly nontrivial problems), there may be practical limitations to improving the performance of the Nystr\"om method simply by increasing the number of eigenfunctions employed. From these numerical results, we conclude that the kernel smoothing method (which requires much less computational effort) is an effective alternative when an accurate estimation of the eigenfunctions is not available.   

 \begin{figure} 
 {\ \centering
\begin{tabular}{cc}
(a) & (b)  \\
\includegraphics[width=.49\linewidth]{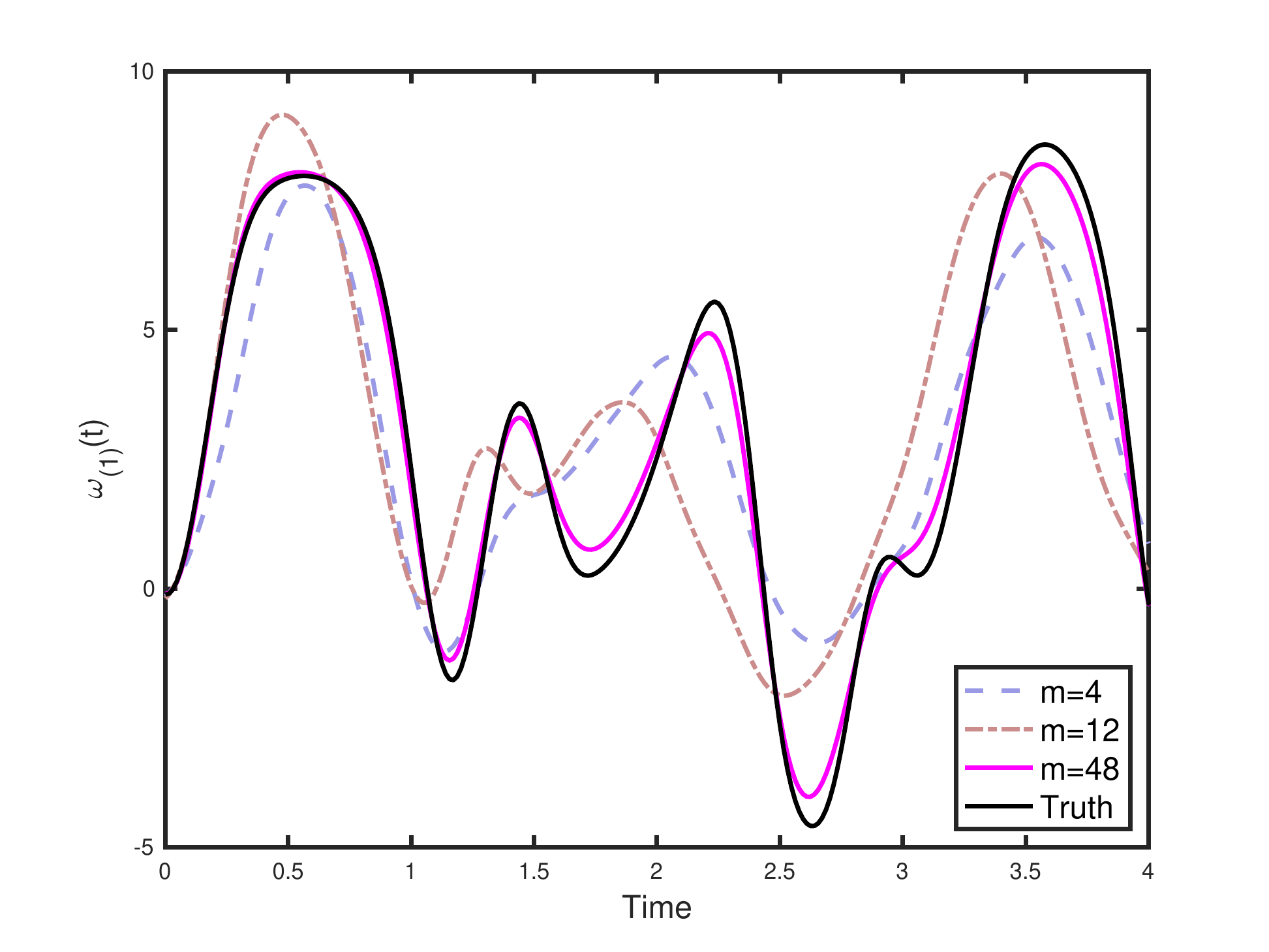}
&
\includegraphics[width=.49\linewidth]{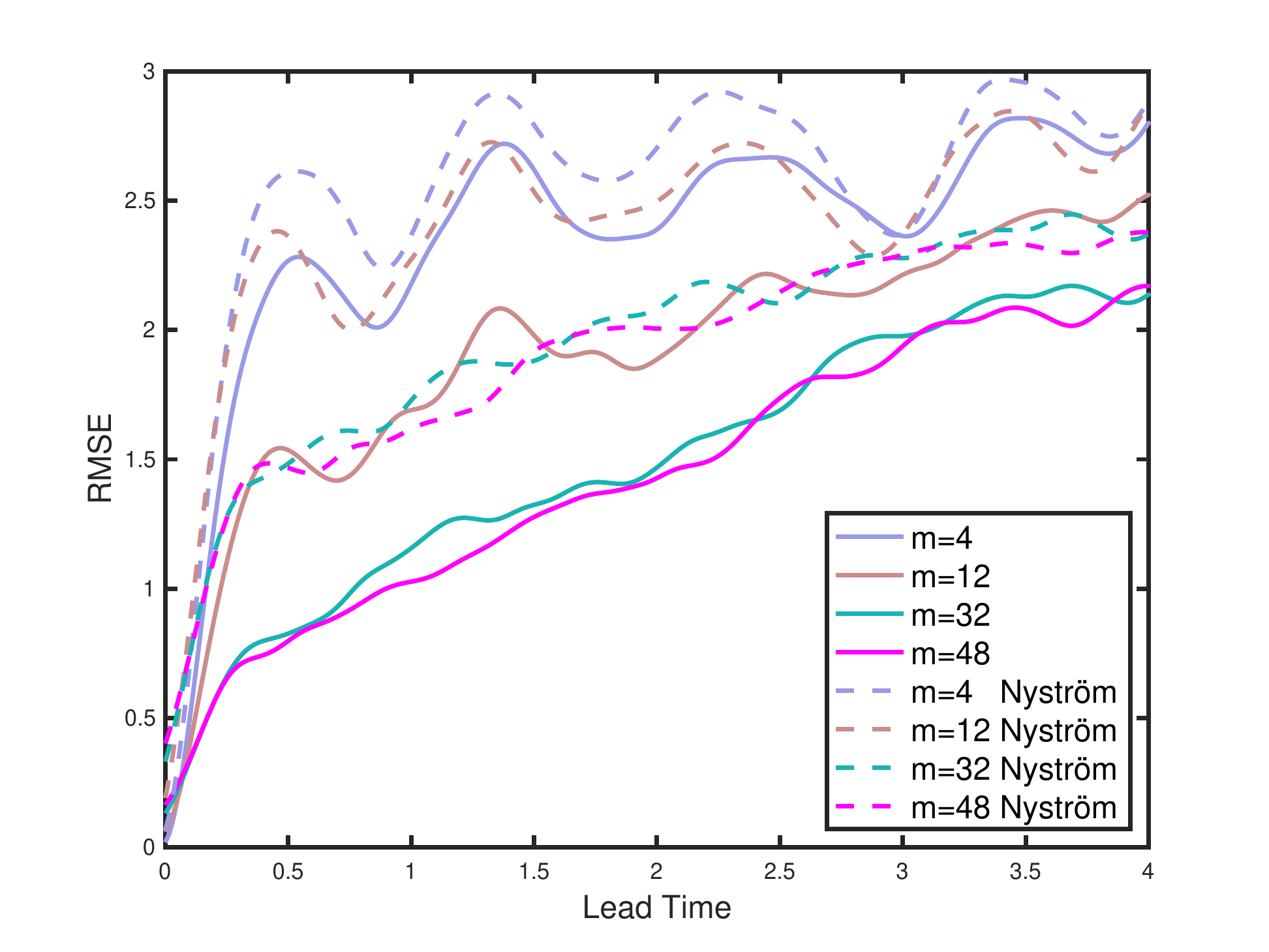}
\end{tabular}%
}
\caption{5D Lorenz-96 example:  (a) Predicting a particular out-of-sample trajectory of the $5$-d Lorenz-96 system with $F=8$ via the kernel smoothing estimate of the conditional expectation using $m=4,12,$ and $48$ past observations. The training and testing data were both drawn from the invariant distribution. (b) The empirical RMSE of the kernel smoothing estimates with $m$ past observations, calculated over $10,000$ out-of-sample points. (c) The Empirical RMSE of both the Nystr\"om and the kernel smoothing estimate  for $m=4, 12, 32$ and $48$.}
\label{L96invPredict}
\end{figure}

 \section{Smoothing and predicting with noisy data}\label{sec4}

To apply the prediction framework discussed in preceding section to real applications, one has to take into account that the available data set is most likely corrupted by noise. To overcome this issue, we design a nonparametric state estimation method, which we will call a smoother, to denoise the data. We should point out that we adopted the terminology smoother since the object of interest is the conditional expectation of the classical Bayesian smoothing problem \cite{sarkka2013bayesian}, which is different than the kernel smoothing method described in Section~\ref{sec:regression}. In Section~\ref{subsection41}, we describe a nonparametric smoother, formulated using the nonparametric regression framework discussed in Section~\ref{sec:regression}. Subsequently, in Section~\ref{subsection42}, we numerically verify the prediction skill of the framework in Section~\ref{sec:regression} where the estimator is trained using the smoothed data obtained from the method in Section~\ref{subsection41}.

\subsection{A nonparametric smoother}\label{subsection41}

We consider smoothing noisy time series observations of the form

\BEA
z_t=x_t+\theta_t,  \quad\quad t=1,\ldots,N,\nonumber
\EEA
where $\theta_t:=\Theta_t(\alpha)$ are realizations of a centered random variable $\Theta_t: A\to\mathcal{X}$ that is independent of $x_t$.  As in Section~\ref{sec:prediction},  $x_t=X_t(\omega)$ where $X_t:\Omega\rightarrow \mathcal{X}$, $\Omega\subset \mathbb{R}^n$, $\mathcal{X} = \mathbb{R}^p$ and $p \leq n$. The goal here is to construct the smoother $\mathbb{E}[X_k\mid z_1,\ldots,z_{m}]$, where $0 \leq k\leq m-1$, and use it as an estimator for $x_k$. In this application, we assume that the full vector $\omega_0$ of initial conditions is drawn from an invariant density. To pose this smoothing problem in the regression framework presented in Section~\ref{sec:regression}, we take the covariate space $\mathcal{Z}_m = (\mathbb{R}^p)^m$ to be the range of the covariate mapping $\mathbb{Z}_m: A\times\Omega \rightarrow \mathcal{Z}_m$ given by $\mathbb{Z}_m(\alpha,\omega):=\{Z_1(\alpha,\omega),Z_2(\alpha,\omega),\ldots,Z_{m}(\alpha,\omega)\}$, where $Z_k(\alpha,\omega)=X_k(\omega)+\Theta_k(\alpha)$. We also let the response space $\mathcal{X}_k$ to be the range of $X_k=\mathbb{R}^p$. Then the least squares estimator is given by $\mathbb{E}[X_k\mid \mathbb{Z}_m]$. If the distribution of $x_k$ is invariant, then the covariate space $\mathcal{Z}_{m}$ and the response space $\mathcal{X}_k$ do not depend on time so the optimal estimator can be trained once using training data drawn from the invariant density.

In general, the noise-free time series $x_t$ is not available for training. Given this constraint, we consider estimating $\mathbb{E}[Z_k\mid \mathbb{Z}_m]$ instead. Denote the available noisy data by $\vec{z}=\{z_1,\ldots,z_{N}\}$ and let $\vec{x}=\{x_1,\ldots, x_{N}\}$,  and $\vec{\theta}=\{\theta_1,\ldots, \theta_{N}\}$ be the uncorrupted data, and the noise, respectively. We employ the VBDM algorithm to obtain the basis functions $\hat{u}_{j,N}(z_{i+1},\ldots,z_{i+m}) = \hat\phi_{j,N}(\alpha_i,\omega_i)$, for $i=1,\ldots,N$ and then represent $E[Z_k\mid \mathbb{Z}_m]$ as a superposition of these basis functions. We motivate the construction of this conditional expectation by noting that the diffusion maps algorithm is robust to low noise perturbations (see Criterion~5 of \cite{cl:06}); more specifically, the error in the spectrum of the graph Laplacian can be controlled as long as the size of perturbation $|\theta_t |< \sqrt{\epsilon}$. Assuming that this argument holds for the estimation of the eigenvectors, we can reasonably expect that $\hat{u}_{j,N}(z_{i+1},\ldots,z_{i+m})= \hat\phi_{j,N}(\alpha_i,\omega_i) \approx  \phi_{j,N}(\omega_i)= u_{j,N}(x_{i+1},\ldots,x_{i+m})$ when the noise size is smaller than $\sqrt{\epsilon}$. 

For a particular class of dynamical systems, the noise robustness of the graph-theoretic techniques can be considerably strengthened by performing delays. In \cite{Giannakis19}, it was shown that if the Koopman operator of the dynamical system on $ \Omega $ has a pure point spectrum, and the noise is i.i.d.\ with finite first four moments, the pointwise estimator for the graph Laplacian determined from the noisy data can be made to agree with the noise-free estimator at any desired tolerance by increasing the embedding window length $m$. Here, we do not assume that the dynamics has pure point spectrum, so the estimates in \cite{Giannakis19} do not necessarily apply, but we can heuristically deduce that, 
\BEA
\langle \vec{z},\hat\phi_{j,N} \rangle_{L^2(\hat{\mu}_N)} &:=& \frac{1}{N} \sum_{i=1}^{N} z_i \hat{u}_{j,N}(z_{i+1},\ldots,z_{i+m}) \nonumber \\
&\approx & \frac{1}{N} \sum_{i=1}^{N} z_i u_{j,N}(x_{i+1},\ldots,x_{i+m}) \nonumber \\
&=& \frac{1}{N} \sum_{i=1}^{N} x_i u_{j,N}(x_{i+1},\ldots,x_{i+m}) + \frac{1}{N} \sum_{i=1}^{N} \theta_i u_{j,N}(x_{i+1},\ldots,x_{i+m})  \nonumber\\
&=& \langle \vec{x},\phi_{j,N} \rangle_{L^2(\mu_N)},\nonumber
\EEA
due to the fact that $\theta_i$ is independent of $x_i$. Here, $\hat{\mu}_N = \sum_{i=1}^N \delta_{\alpha_i,\omega_i}/N$ is the discrete sampling measure on noisy data, whereas $\mu_N$ is the discrete sampling measure on uncorrupted data. This suggests that we can approximate the smoother $\mathbb{E}_{L,N}[X_k \mid \mathbb{Z}_m]$ as, 
\BEA
\mathbb{E}_{L,N}[X_k \mid \mathbb{Z}_m] := \sum_{j=0}^L \langle \vec{x},\phi_{j,N} \rangle_{L^2(\mu_N)} \phi_{j,N}  \approx \sum_{j=0}^L \langle \vec{z},\hat\phi_{j,N} \rangle_{L^2(\hat\mu_N)} \phi_{j,N} = \mathbb{E}_{L,N}[Z_k \mid \mathbb{Z}_m].\nonumber
\EEA
Here, the approximation is due to the fact that the construction of the estimator is based solely on the noisy data. A more detailed error analysis is an open problem that is beyond the scope of this paper. Note that the kernel smoothing estimator discussed in Section~\ref{sec::kernelsmooth} does not possess the robustness-to-noise property that the VBDM basis does and furthermore the kernel smoothing estimate of $z_k$ is not an approximation of the kernel smoothing estimate trained using $x_k$. Thus, an application of the kernel smoothing estimator discussed in Section~\ref{sec::kernelsmooth} trained solely on noisy data would not yield a good approximation to the desired smoother, $\mathbb{E}[X_k\mid \mathbb{Z}_m]$.

In the remainder of this subsection, we will demonstrate the effectiveness of this smoother in recovering $x_k$ from an out-of-sample sequence $(z_1,\ldots,z_m)$ of noisy observations.  We will show numerically that choosing $m>1$ reduces the RMSE when estimating $x_k$. We will also demonstrate the sensitivity of the RMSE as the parameter $k$ is varied. 
To verify the accuracy of the proposed smoother, we compare the RMSE of this method to  the RMSEs of the Ensemble Kalman Filter \cite{evensen:94} and 4D-Var \cite{lorenc:86}, both of which are very popular data assimilation methods that are operationally used in weather forecasting. 

{\bf Smoothing noisy observations of the Lorenz-63 system:} Consider the Lorenz-63 system given by
\BEA
\frac{d\omega_{(1)}}{dt}&=&\sigma(\omega_{(2)}-\omega_{(1)})\nonumber\\
\frac{d\omega_{(2)}}{dt}&=&\omega_{(1)}(\rho-\omega_{(3)})-\omega_{(1)}\label{L63} \\
\frac{d\omega_{(3)}}{dt}&=&\omega_{(1)}\omega_{(2)}-\beta\omega_{(3)} \nonumber
\EEA
with the standard parameters $\sigma=10, \rho=28$ and $\beta=8/3$ that give rise to the famous ``butterfly"-like attractor \cite{lorenz:63}.
In this example, we are interested in estimating $\omega_{(1)}$, from observations of the form $Z_t(\omega(0))=\omega_{(1)}(t)+\theta_t$, and we will use the VBDM algorithm to construct the smoother $\mathbb{E}_{L,N}[Z_k\mid \mathbb{Z}_m]$. Practically, given an out-of-sample sequence $\mathbf{z}^{out}_{t,m}=(z^{out}_{t+1},\ldots,z^{out}_{t+m})$ of consecutive noisy observations, we use the Nystr\"om method to evaluate $\mathbb{E}_{L,N}[X_{t+k}|\mathbf{z}^{out}_{t,m}]$, where $k=1,\ldots,m$ and use this as an estimator for $x^{out}_{t+k}$.  To smooth a long sequence of noisy observations, we independently apply the constructed conditional expectation on each $t$ and the corresponding $\mathbf{z}^{out}_{t,m}$ sequence of the trajectory. \comment{We note that with this parallel smoothers, we cannot recover the first $k-1$ and the last $m-k$ observations in the trajectory in a high-degree of accuracy, as we shall see.} In all of the following numerical experiments, the training data $\mathbb{Z}_m$ is constructed by taking sequential $m$ observations of $\omega_{(1)}(t)$ and adding $\theta_t$ to each of these elements.

In the first experiment, we use $N=12,000$ observations of $z_t$ with $\theta_t\sim \mathcal{N}(0,4)$ and construct the conditional expectation estimators, $\mathbb{E}_{L,N}[Z_k|\mathbb{Z}_m]$, with $m=5$ and $k=1,\ldots,5$ using $L=120$ eigenfunctions. The observation time step for $z_t$ is $\Delta t=0.1$.  We evaluate this smoothing operator on an out-of-sample trajectory of length $N_{out}=10,000$, corrupted by four noise types with variance approximately $4$ : (1) Gaussian noises $\mathcal{N}(0,4)$, (2) Student's-t noises with $8/3$ degrees of freedom, (3) Uniformly distributed noises over $(-\sqrt{48}/2,\sqrt{48}/2)$, and (4) Time varying noises of the type $2\sin(tU)$, where $U$ is uniformly distributed over $[-1/2,1/2]$. In each of the following experiments, the basis is constructed using data corrupted by $\mathcal{N}(0,4)$ noise. This choice of standard deviation, 2, is roughly 25\% of the climatological standard deviation, $7.9246$. Table~\ref{Lftablegauss} shows the RMSEs of the smoothers $\mathbb{E}[Z_k|\mathbb{Z}_5]$ when the observed component is corrupted by$\mathcal{N}(0,4)$ noise for $1 \leq k \leq 5$. From Table~\ref{Lftablegauss}, one can see that the smallest error is obtained for $k=2$. Table~\ref{Lftablenongauss} shows the RMSEs of the smoother $\mathbb{E}[Z_2\mid \mathbb{Z}_5]$ in the cases when the observed component is corrupted by the four noise types mentioned previously. In Figure~\ref{L63fnoise}, we show the smoothed trajectories compared to the truth and noisy observations. We should point out that the $k=2$ smoother forces us to discard the first and the last $N_{out}-(m-2)$ observations in the trajectory. Note that the RMSEs shown in each of these tables are the errors in recovering $\omega_{(1)}$, computed by averaging the errors of time indices $k$ to $N_{out}-(m-k)$. 

As seen in Table~\ref{Lftablenongauss}, the smoother, $\mathbb{E}_{L,N}[Z_2|\mathbb{Z}_{5}]$ performs better than an ensemble Kalman filter with $64$ ensemble members constrained to observing the same one-dimensional noisy component  $\omega_{(1)}$ used for training the smoother (denoted by 1 observation in the table). We also found that the proposed smoother is more accurate than 4D-Var constrained to using only one observation $(\omega_{(1)}$ only) or two observations (both $\omega_{(1)}(t)$ and $\omega_{(2)}(t)$). In this table, for diagnostic purpose, we also report the RMSEs of EnKF and 4D-Var when noisy observations of all three components are available. One can see that, in this case, 4D-Var is superior; however, when only one component is observed, the proposed smoother (which requires no knowledge of the dynamics) is more accurate than both EnKF and 4D-Var. We should point out that in these numerical experiments the 4D-Var is implemented with $m=5$ so that the configuration is similar to that of the non-parametric smoother. 
 
Finally, we note here that if we instead use the kernel smoothing estimator described in Section~\ref{sec::kernelsmooth} to approximate the smoother $\mathbb{E}[Z_k\mid \mathbf{Z}_m]$, we find that the RMSE of smoothing the same observations with Gaussian noise (as in Table~\ref{Lftablenongauss}) is $1.5605$. Therefore, while the smoothed time series are less noisy than the observed data, the kernel smoothing based smoother performs worse than the EnKF with only $1$ observation.

 \begin{figure}
 {\ \centering
\begin{tabular}{cc}
(a) & (b)  \\
\includegraphics[width=.5\linewidth]{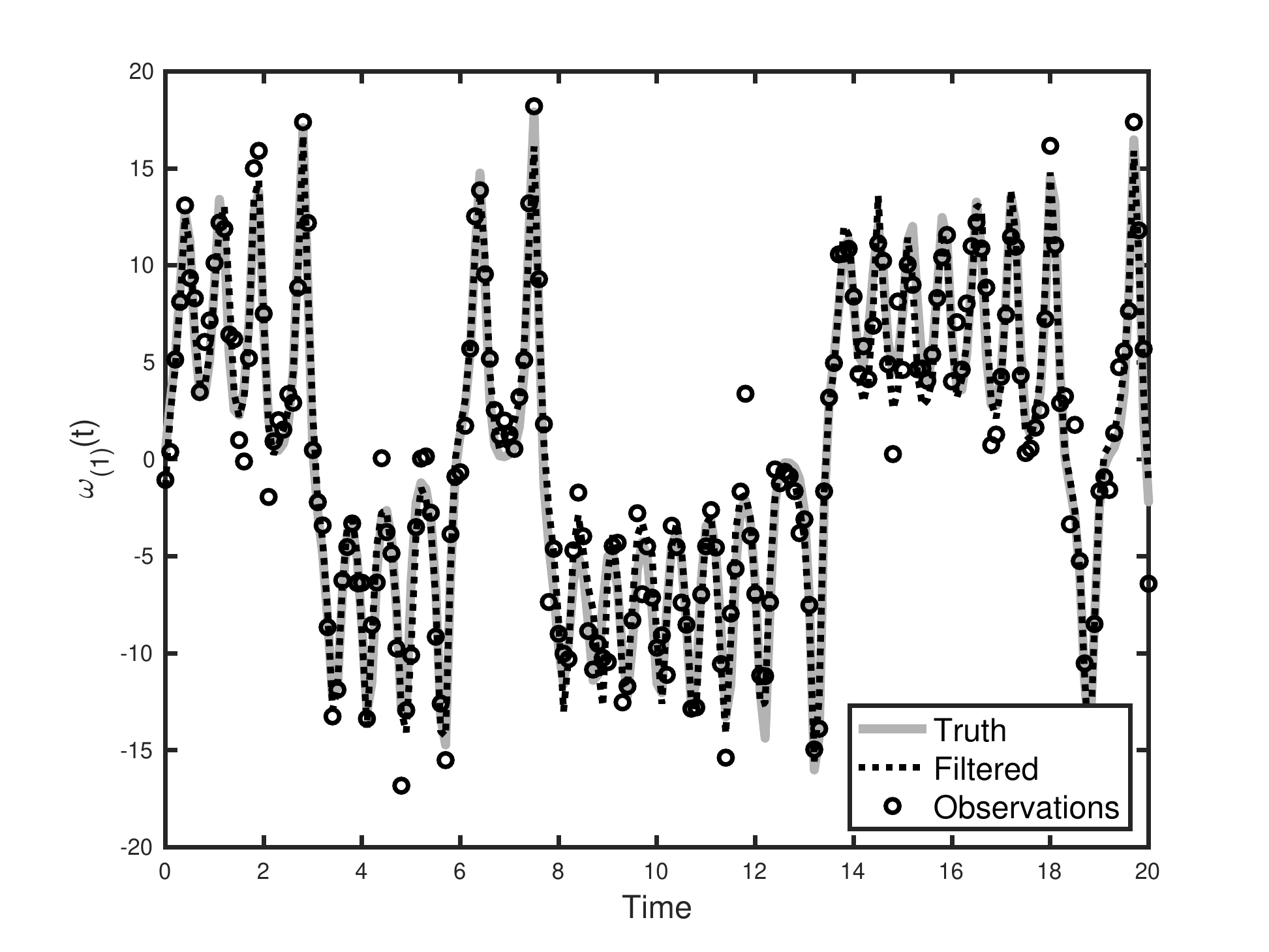}
&
\includegraphics[width=.5\linewidth]{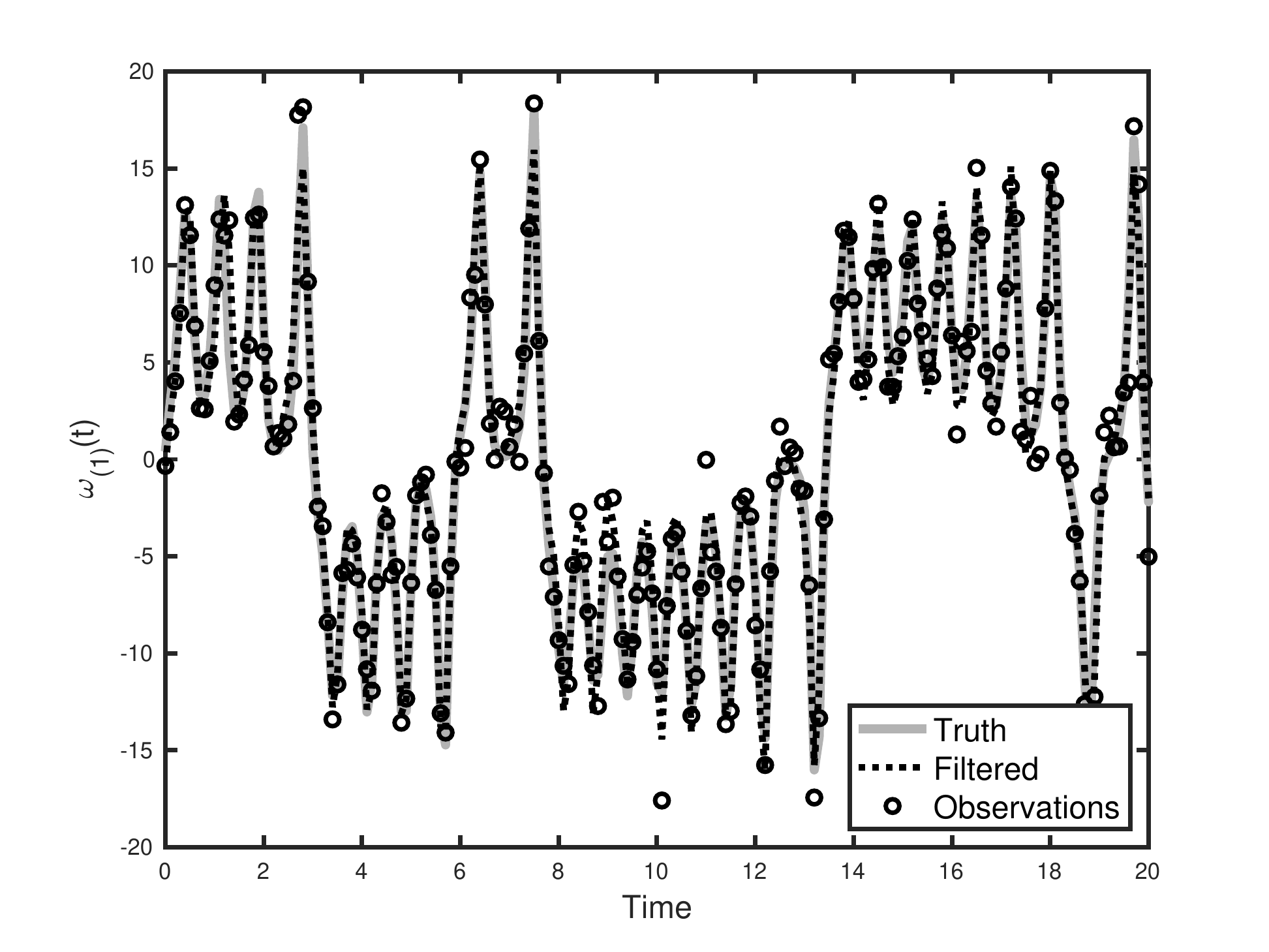}
\\
(c) & (d)  \\
\includegraphics[width=.5\linewidth]{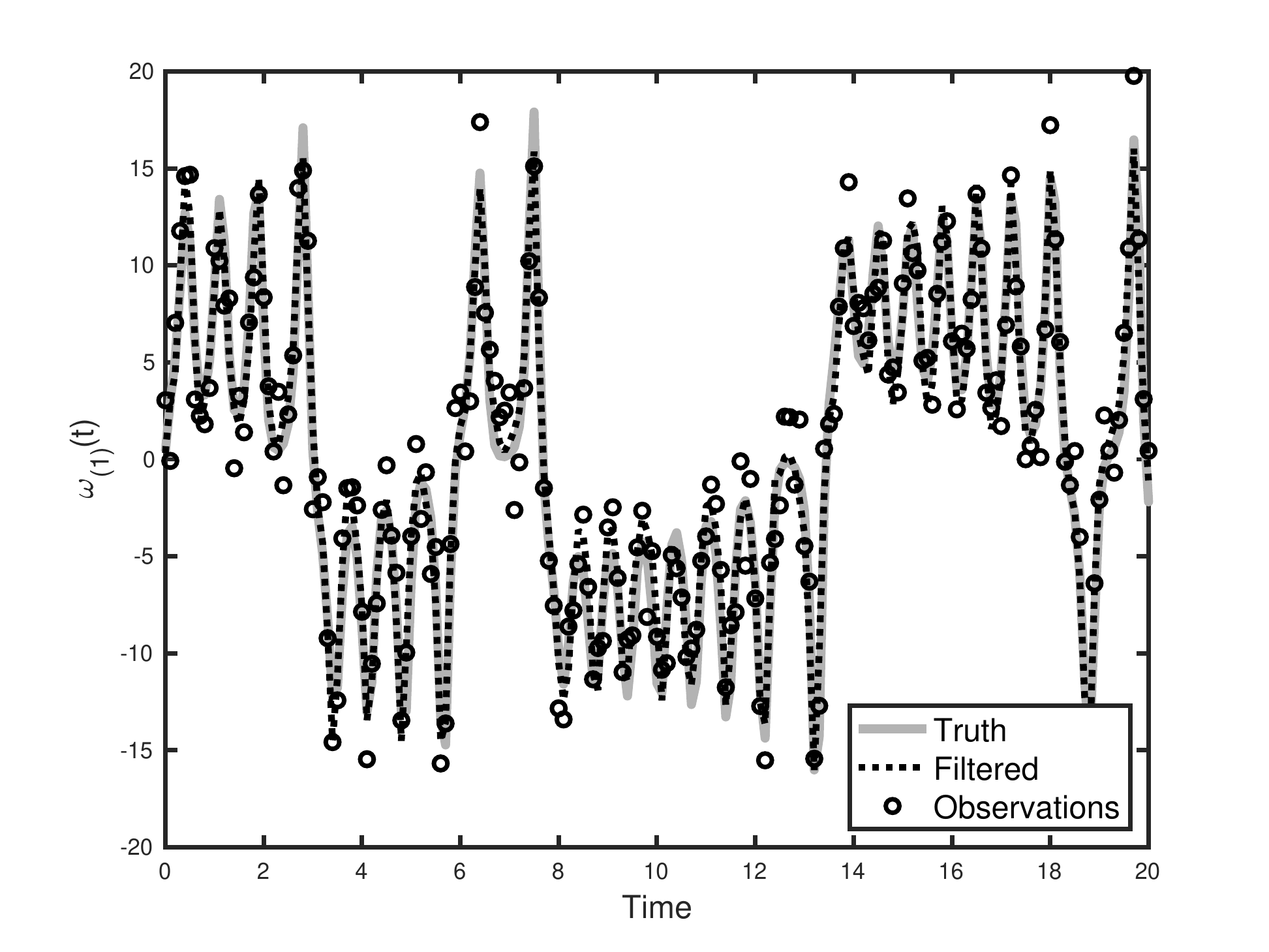}
&
\includegraphics[width=.5\linewidth]{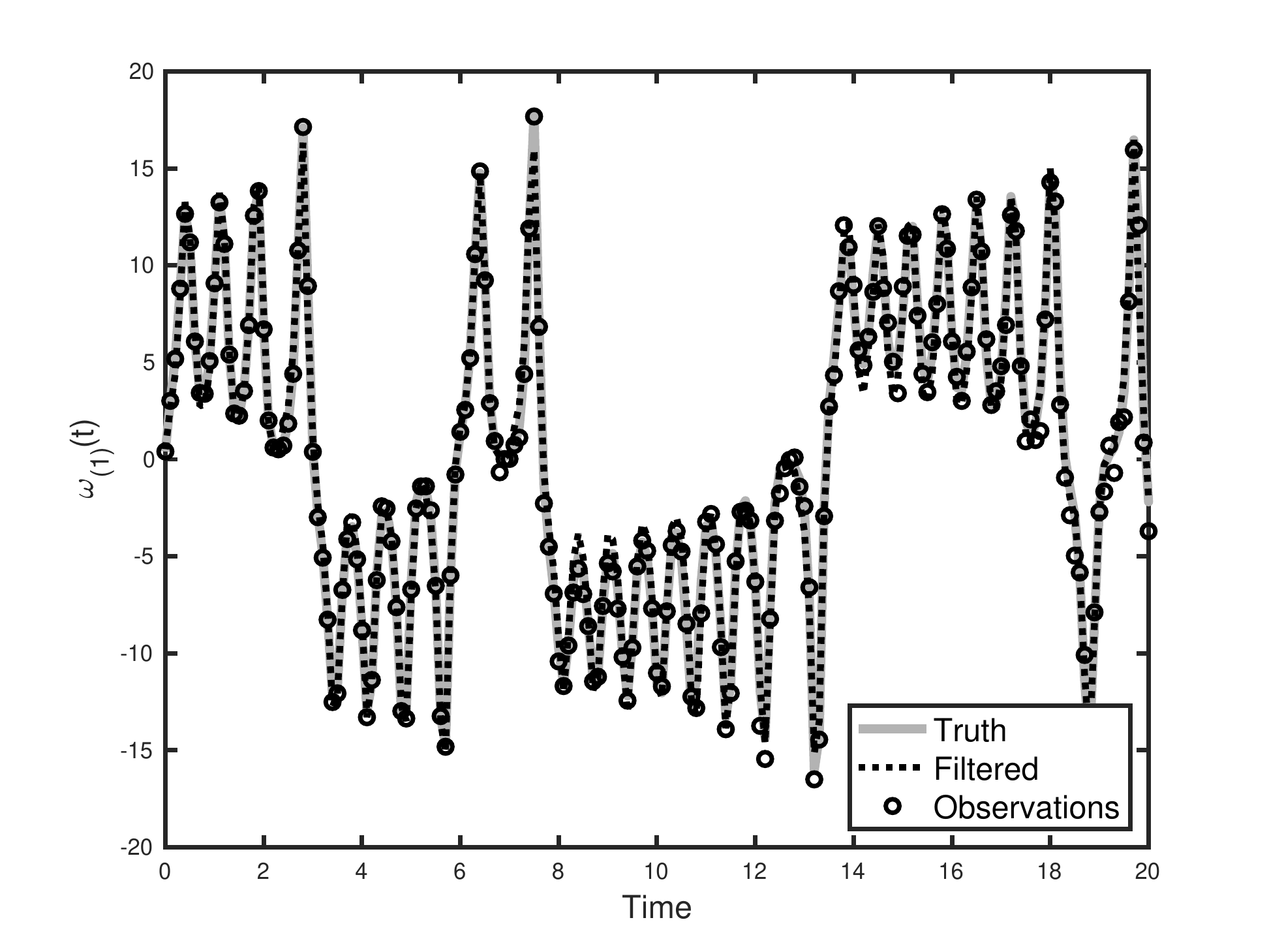}
\end{tabular}%
}
\caption{Smoothed trajectories of $\omega_{(1)}$, compared to the truth and noisy observations. Here, the same smoother, constructed using data corrupted by Gaussian noise $\mathcal{N}(0,4)$, was used to recover out-of-sample data corrupted by: (a) Gaussian noises $\mathcal{N}(0,4)$; (b) Student's-t noises with $8/3$ degrees of freedom; (c) Uniformly distributed noises over $(-\sqrt{48}/2,\sqrt{48}/2)$; and (d) Time varying noises of the type $2\sin(tU)$, where $U$ is uniformly distributed over $[-1/2,1/2]$. See Tables~\ref{Lftablegauss} and \ref{Lftablenongauss} for the RMSEs.  }
\label{L63fnoise}
\end{figure}

\begin{table} 
\centering
\begin{tabular}{ |c| c| c| c| c| c|  }
\hline
 k & 1& 2& 3 & 4& 5\\
  \hline
 RMSE & 1.3591& 1.0663&    1.1243&    1.2394&    1.4928\\
 \hline
\end{tabular}
\vspace{8pt}
\caption{RMSEs of the smoother for the $\omega_{(1)}$ component of the Lorenz-63 subjected to i.i.d  $\mathcal{N}(0,4)$ noise. The  smoother $\mathbb{E}_{L,N}[Z_k|\mathbb{Z}_m]$ was constructed using $N=12,000$ training data points, consisting of $m=5$ sequential observations and $L=120$ eigenfunctions. Each RMSE is averaged over out-of-sample data points with indices $k$ to $N_{out-(m-k)}$, where $N_{out}=10,000$ data points.}
\label{Lftablegauss}
\end{table}

\begin{table} 
\centering
\begin{tabular}{ |c| c| c| c| c| c|  }
\hline
  &Gaussian& Student's-t& Uniform& Time varying \\
  \hline
  Smoother m=5, k=2& 1.0663& 1.1339& 1.1406&    1.1455\\
 \hline
EnKF (1 obs) &1.3439& 1.6589 &  1.4055& 1.3704\\
\hline
EnKF(2 obs) &.7435& 0.7871&  .7282& .7538 \\
\hline
EnKF(3 obs) &.6343& .7806&  .6242 & .6674 \\
\hline
4D-VAR (1 obs) &2.0935& 2.2131 &  2.3187& 2.3145\\
\hline
4D-VAR (2 obs) &1.4971& 1.4099 &  1.6088& 1.7032\\
\hline
4D-VAR (3 obs) &.5436& .6437&  .5198& .7785\\
\hline
\end{tabular}
\vspace{8pt}
\caption{ RMSEs of the smoother for the $\omega_{(1)}$ component of the Lorenz-63 model subjected to: (1) Gaussian noises $\mathcal{N}(0,4)$, (2) Student's-t noises with $8/3$ degrees of freedom, (3) Uniformly distributed noises over $(-\sqrt{48}/2,\sqrt{48}/2)$, and (4) Time varying noise of the type $2\sin(tU)$, where $U$ is uniformly distributed over $[-1/2,1/2]$. The smoothing operator was constructed, as in Table~\ref{Lftablegauss}, by using $N=12,000$ training data points consisting of $m=5$ sequential observations, corrupted by Gaussian noise and $L=120$ eigenfunctions. The RMSEs were also calculated the same way as in Table~\ref{Lftablegauss}. The same underlying trajectory was used for all of the RMSE calculations.  The last six rows report benchmark results of applying the EnKF with $64$ ensemble members and 4D-Var with $m=5$, observing one to three noisy components, respectively.}
\label{Lftablenongauss}
\end{table}

{\bf Smoothing noisy observations of the Lorenz-96 system:} We consider the proposed smoother for $k=2$ to estimate the first component $\omega_{(1)}(t)$ of the $40$-dimensional Lorenz-96 system \eqref{L96model} with $F=8$. In this numerical experiment, the training data consist of observations of $\omega_{(1)}(t)$ perturbed by Gaussian noise, $\mathcal{N}(0,1)$ but, as in Lorenz-63 example, we apply the filter to out-of-sample trajectories with various noise types with variance close to $1$. This choice of noise standard deviation, 1, is roughly 25\% of the climatological standard deviation, $3.5868$.

For brevity, we only report the results for $m=6$, $k=3$, $L=200$, and $N=12,000$ training data, which we verified on three out-of-sample trajectories, each of length $N_{out}=10,000$ sequential observations.  These three trajectories were respectively generated by perturbing a single out-of-sample trajectory of length $N_{out}$ with $\mathcal{N}(0,1)$ noise, Student's-t noise with $10$ degrees of freedom, and noise drawn uniformly from $[-1.8,1.8]$. The training and testing time series were generated using RK4 with an observation time step of  $0.05$. See Table~\ref{L96table} for the RMSEs as well as benchmark results of applying EnKF and 4D-VAR in the same noise regimes. Note that the proposed smoother, $\mathbb{E}_{L,N}[Z_3\mid\mathbb{Z}_6]$, performs better than both EnKF and 4D-Var even when the these two schemes were allowed to observe $30$ noisy components.

\begin{table}
\centering
\begin{tabular}{|c|c|c|c|}
\hline
& Gaussian& Student's-t & Uniform\\
\hline 
Smoother $m=6, k=3$ & $.5958$&$.6100$&$ .5968$ \\
\hline
EnKF (10 obs) & $.7234$ & $0.7387$& $.7151$\\
\hline
EnKF (30 obs) & $.6229$ & $.6785$& $.6416$\\
\hline
EnKF (40 obs) & $.2492$ & $.2428$& $.2211$\\
\hline
4D-VAR (10 obs)&$2.4722$&$2.5051$&$2.4691$\\
\hline
4D-VAR (30 obs)&$2.0373$&$2.0764$&$2.0141$\\
\hline
4D-VAR (40 obs)&$.2933$&$.2983$&$.2822$\\
\hline

\end{tabular}
\vspace{8pt}
\caption{RMSEs of the smoother for the $\omega_{(1)}$ component of the Lorenz-96, with $F=8$, subjected to: (1) Gaussian noise $\mathcal{N}(0,1)$; (2) Student's-t noise with $10$ degrees of freedom; (3) uniform noise drawn from $[-1.8,1.8]$. The smoothing operator for $k=3$ was constructed using $N=12,000$ training data points consisting of $m=6$ sequential observations constructed by Gaussian noise and $L=200$ eigenfunctions. The RMSEs are calculated using the same underlying trajectory of length $N_{out}=10,000$. The last six rows report benchmark results of applying the EnKF with $64$ ensemble members and 4D-Var with $m=6$, observing 10-40 noisy components, respectively.}
\label{L96table}
\end{table}

Finally, we consider a more chaotic regime with $F=16$. We construct an estimator to smooth observations of $\omega_1(t)$ corrupted with Gaussian noise, $\mathcal{N}(0,1)$. The smoother is trained with parameters $L=300, N=20,000$ and $m=6$ using a set of training data corrupted by Gaussian noise, $\mathcal{N}(0,1)$. The RMSE is computed over  $N_{out}=10,000$ verification data which was generated by perturbing the true $\omega_{(1)}$ component of a trajectory  by Gaussian noise,$\mathcal{N}(0,1)$. 

In Table~\ref{LF16}, we report the RMSEs of the smoothing estimates using various choices of $1\leq k \leq 6$. Figure~\ref{LF16fig} shows the results of the smoother with $k=3$. As a benchmark, we report the results of applying  EnKF with $64$ ensemble members and 4D-Var with $m=6$, observing $10$, $30$ and $40$ components of the system in Table~\ref{LF16enkf}. Note that the estimator, which is solely constructed using noisy data $\omega_{(1)}$, with $k=3$ performs better than both the EnKF and 4D-Var estimates observing $30$ noisy components. 

\begin{table} 
\centering
\begin{tabular}{ |c| c| c| c| c| c|c|  }
\hline
 k & 1& 2& 3 & 4& 5& 6\\
  \hline
 RMSE & .8768& .6688&    .6543&    .7321& .7946&.8935\\
 \hline
\end{tabular}
\vspace{8pt}
\caption{RMSEs for smoothing the $\omega_{(1)}$ component of Lorenz-96 with $F=16$ subject to i.i.d  $\mathcal{N}(0,1)$ noise. The smoother is constructed using $L=300$ eigenfunctions, estimated from $N=20,000$ training data points consisting of $m=6$ sequential observations. The RMSEs are calculated over a noisy trajectory consisting of $10,000$ data points. The same out-of-sample trajectory was used for each of the RMSE calculations. }
\label{LF16}
\end{table}

 \begin{figure}
 \centering
\includegraphics[width=.5\linewidth]{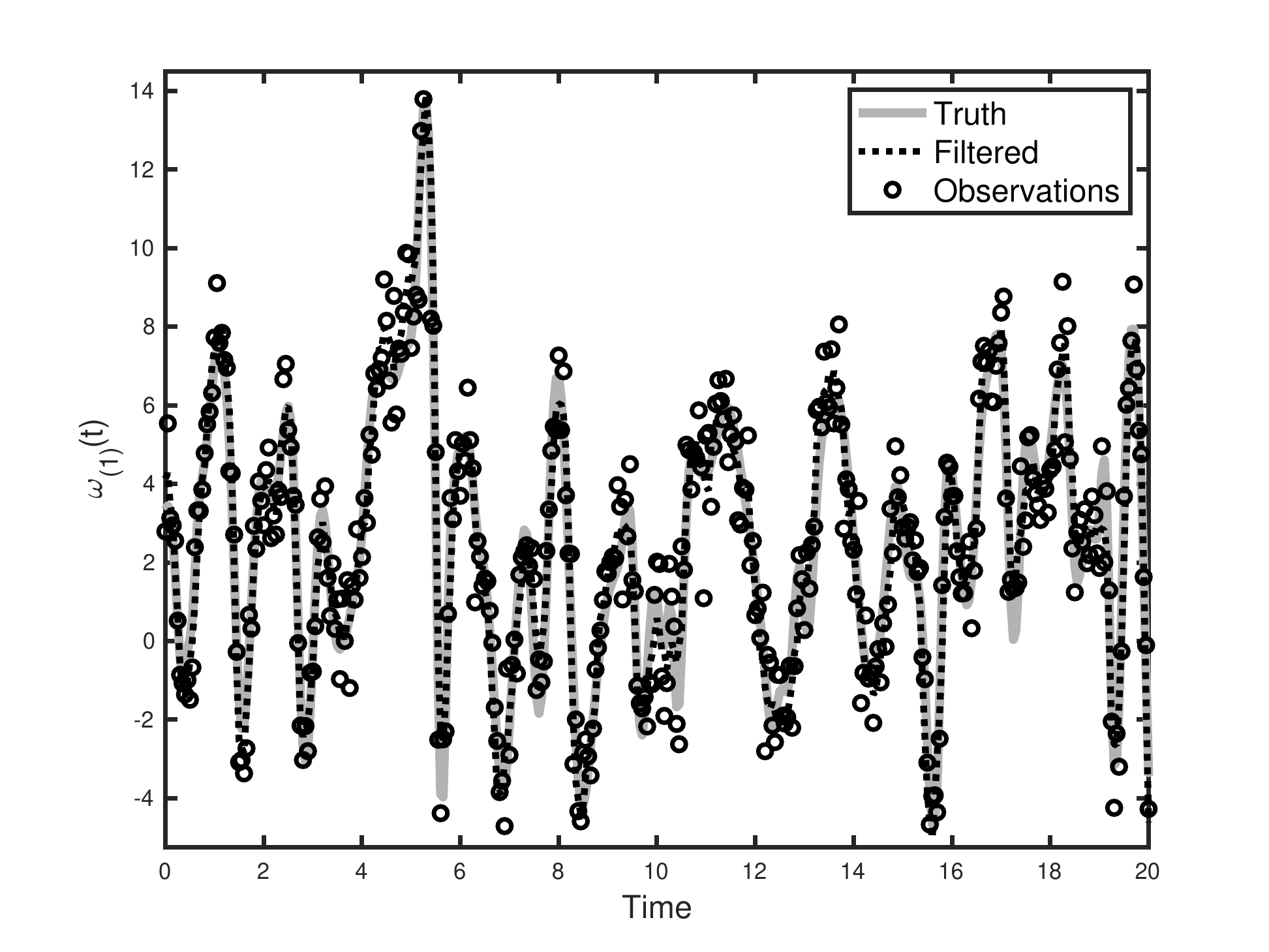}
\caption{ Smoothed trajectories of $\omega_{(1)}$ compared to the truth and noisy observations. The smoother is constructed as described in Table~\ref{LF16} for $m=6, k=3$. The RMSE is $.6543$}
\label{LF16fig}
\end{figure}

\begin{table}
\centering
\begin{tabular}{|c|c|c|c|}
\hline
No. Observation&10 & 30 &40\\
\hline
EnKF &.90558& .8802& .35061\\
\hline
4D-Var $m=6$ &4.1632& 4.052& .36414\\
\hline
\end{tabular}
\vspace{8pt}
\caption{RMSEs for smoothing $\omega_{(1)}$ component of Lorenz-96 with $F=16$ subject to i.i.d  $\mathcal{N}(0,1)$ noise. The configuration here is similar to that in Table~\ref{LF16}.}
 \label{LF16enkf}
\end{table}

\subsection{Prediction with smoothed training data}\label{subsection42}
In the preceding subsection, we introduced a smoother for denoising an out-of-sample sequence of noisy observations. In this subsection, we first denoise the noisy time series data using the smoother and then use the smoothed data as a surrogate for the true data to construct a kernel smoothing estimator as discussed in Section~\ref{sec::kernelsmooth}. 
%This kernel smoothing estimator can be thought of as an approximation of the Koopman operator associated with the underlying model.  
In the following numerical result, we verify the prediction skill of employing these two steps given noisy time series observations of the five-dimensional Lorenz-96 example from Section~\ref{sec:estMZ} and compare it with the true prediction model.

The training data consists of $N=20,000$ sequences, each consisting of sequential observations of the $\omega_{(1)}$ component of the 5D Lorenz-96 model observed at time steps of $\Delta t=1/64$. However, unlike the example in Section \ref{sec:estMZ}, these observations are generated by perturbing the true $\omega_{(1)}(t)$ component by $\mathcal{N}(0,1)$ noise. We also generate an additional $N_{out}=10,000$ sequences as above for verification. We denoise both the noisy training and verification data using the smoother $\mathbb{E}_{L,N}[Z_k|\mathbb{Z}_{m_s}]$ with $k=3$, $m_s=6$ and $L=120$. Here, we defined $m_s$ in place of $m$ to avoid the conflict of notation in the later discussion which refers to $m$ as the memory length in the kernel smoothing estimates. The smoothed training data is then used to construct the kernel smoothing estimator described in Section~\ref{sec::kernelsmooth}. Subsequently, the kernel smoothing estimator is evaluated at each of the smoothed verification data. We should point out that while each smoothed verification data point required $m_s-k=3$ future observations, the numerical experiments below, which evaluate the prediction skill beyond $3\Delta t=3/64$ time units, are still valid prediction tests.

In Figure~\ref{smoothandpredict}(a), we show the kernel smoothing estimates for different memory length, $m$, and the true trajectory, all starting from a particular out-of-sample initial condition. While the prediction is less accurate than the one obtained from the kernel smoothing estimator constructed using noise-free data (see panel (c) for $m=48$), one can still see that the prediction is somewhat improved as the memory length $m$ is increased. This is also confirmed by the RMSEs plot in panel (b). In panel (d), we overlay the RMSEs in panel (b) with those corresponding to the prediction model constructed from the noise-free data (as in Figure~\ref{L96invPredict}). While the model constructed from the noise-free data is more accurate, the discrepancy between the noise-free model and the noisy model decreases as $m$ increases.

\begin{figure}
 {\ \centering
\begin{tabular}{cc}
(a) & (b)  \\
\includegraphics[width=.5\linewidth]{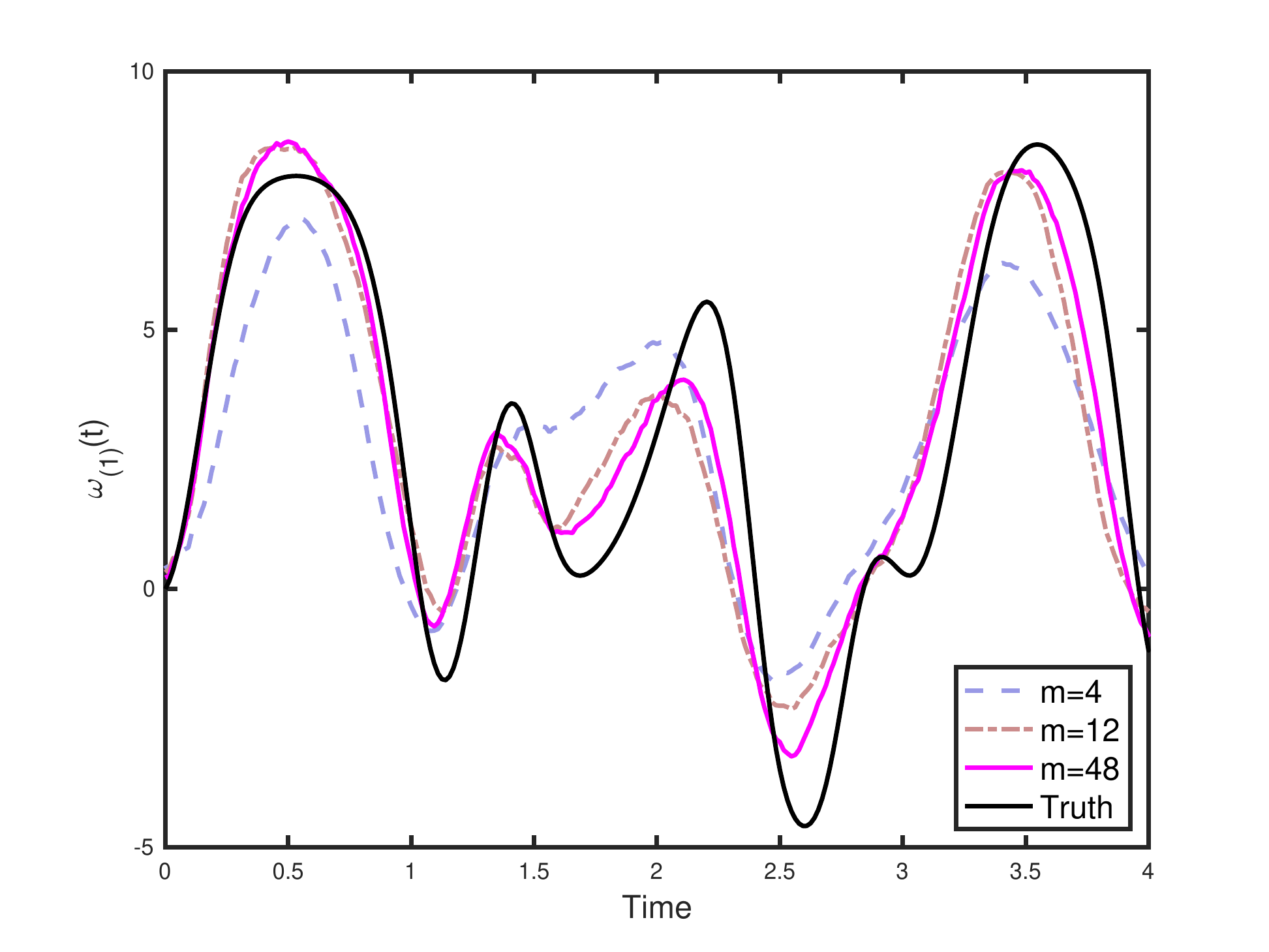}
&
\includegraphics[width=.5\linewidth]{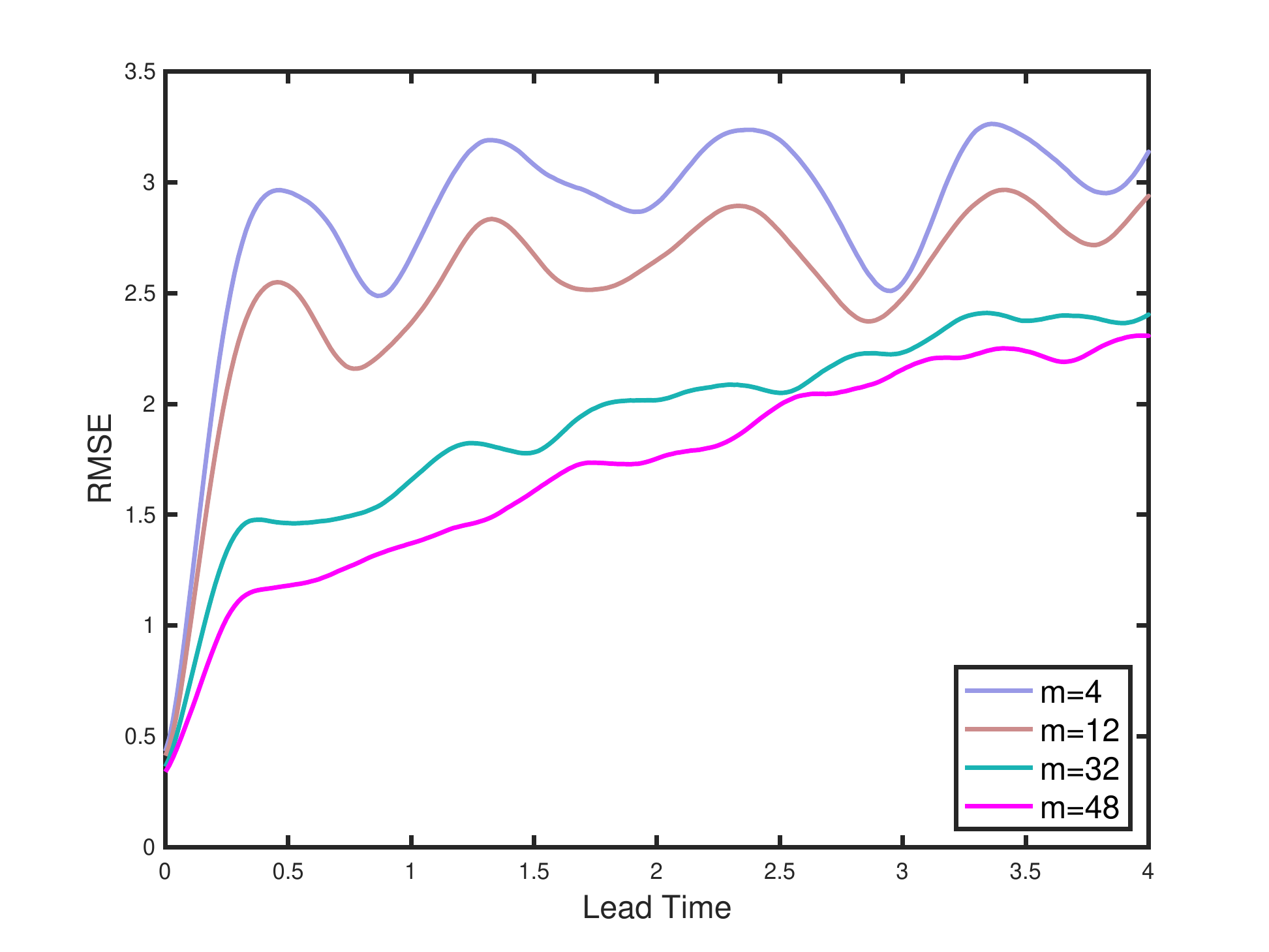}
\\
(c) & (d)  \\
\includegraphics[width=.5\linewidth]{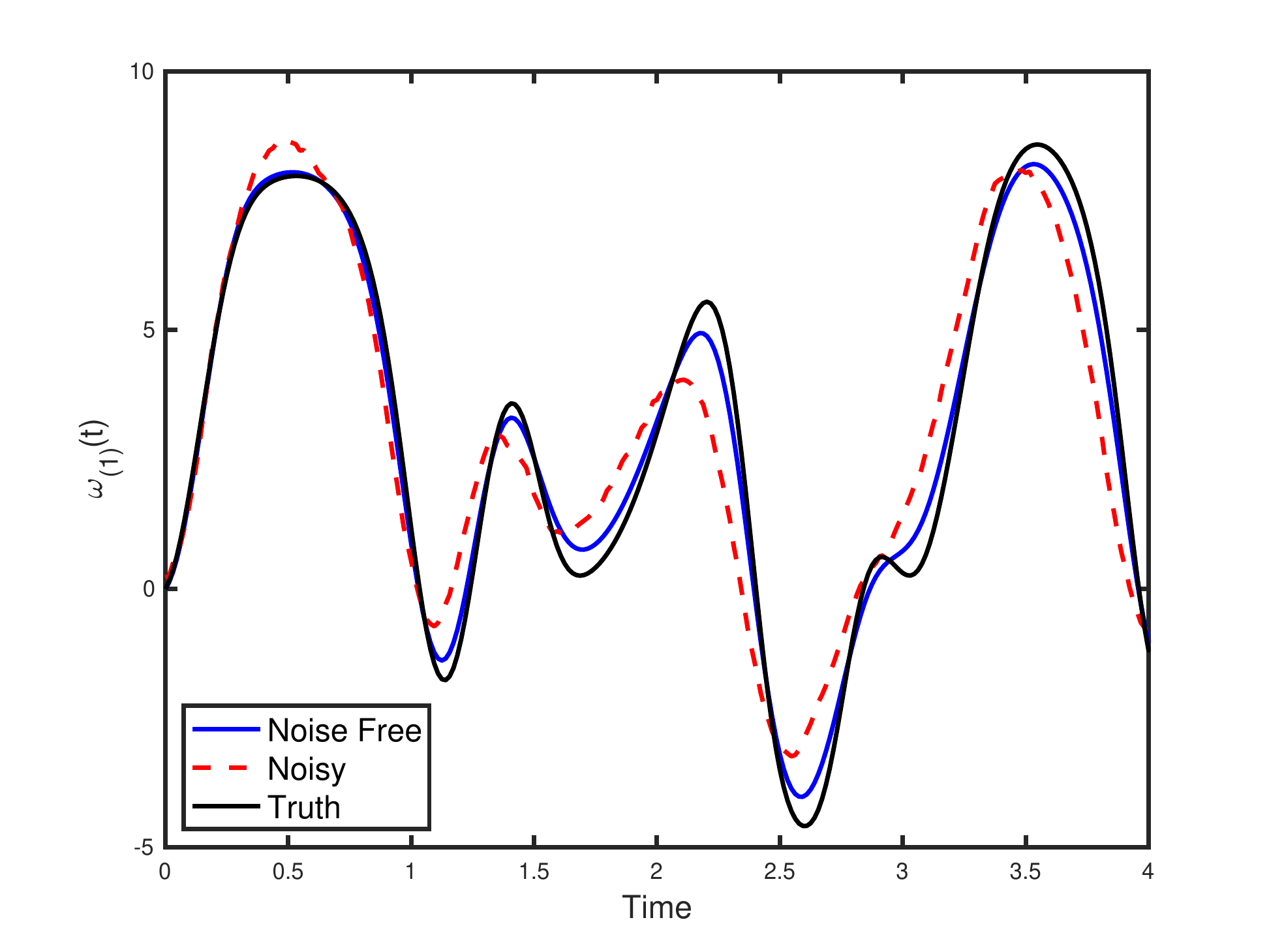}
&
\includegraphics[width=.5\linewidth]{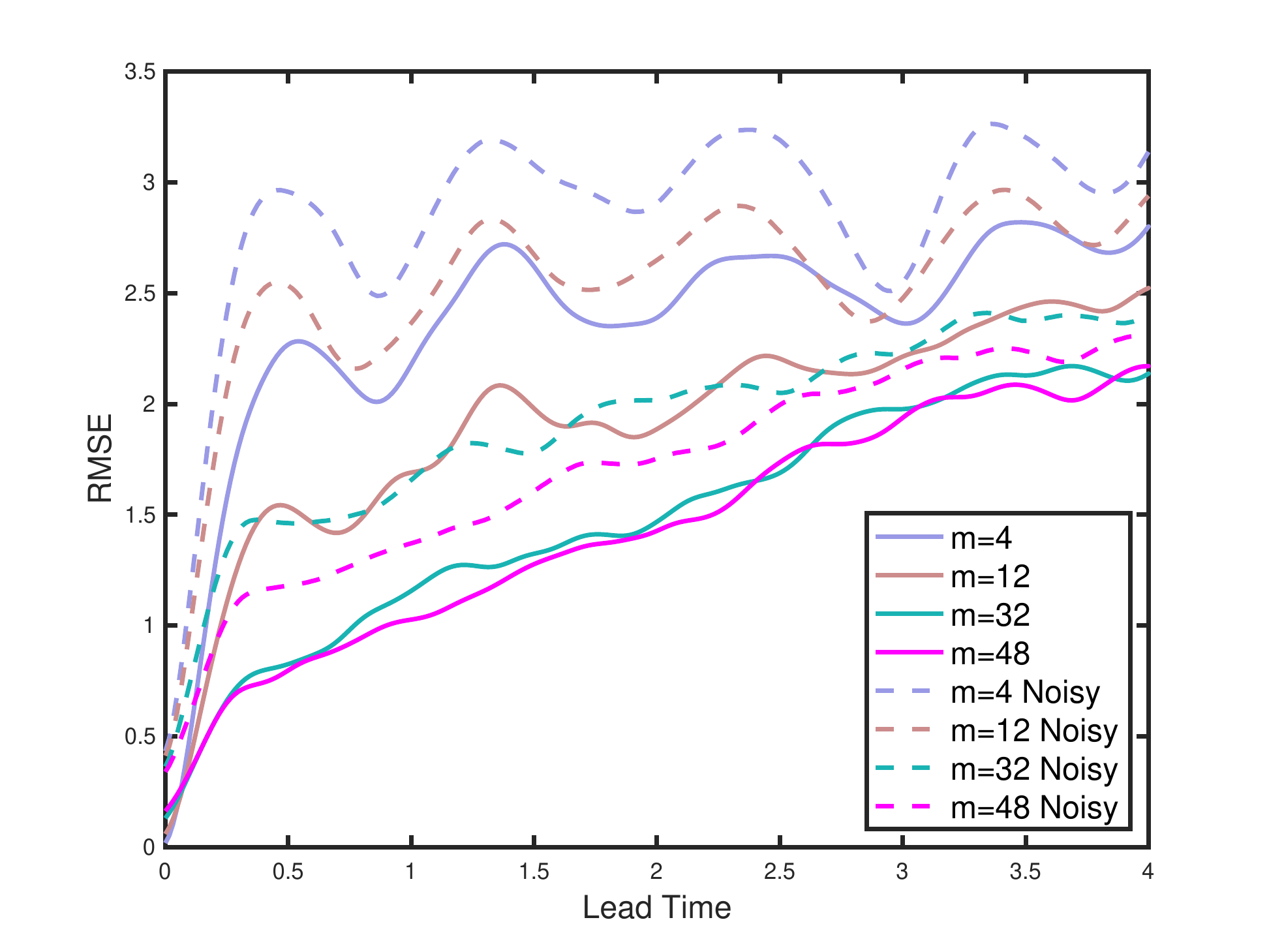}
\end{tabular}%
}
\caption{Result of applying the kernel smoothing estimate on noisy observations of the first component of the 5D Lorenz-96 model with $F=8$. The observations were first denoised using the nonparametric smoother and subsequently used for training the kernel smoothing estimator discussed in Section~\ref{sec::kernelsmooth}. (a) Trajectory of a particular out-of-sample prediction using the smoothed data for various memory lengths; (b) The corresponding RMSEs as functions of lead time for the kernel smoothing estimator constructed using smoothed observations; (c) A comparison of the predicted trajectory against that constructed from the corresponding noise-free data set for $m=48$; (d) The RMSEs in panel (b) overlaid with the corresponding RMSEs obtained from noise-free data experiment (exactly the RMSEs in Figure~\ref{L96invPredict}(b)).}
\label{smoothandpredict}
\end{figure}

\section{Summary}\label{sec5}

In this paper, we have studied aspects of forecasting and denoising of non-Markovian time series generated by partially observed dynamical systems. Within the context of kernel methods, where there is a mature theory on the consistency of empirical estimators, we have explored two distinct approaches, namely projection-based methods using the Nystr\"om out-of-sample extension approach and smoothing methods using Markov integral operators. We refer collectively to these approaches as kernel analog forecasting (KAF) \cite{zhao2016analog}, since in many ways they can be viewed as kernel-based generalizations of the classical analog forecasting approach proposed by Lorenz in 1969 \cite{lorenz1969atmospheric}.   

Previously, the consistency of KAF in the large-data limit was studied from the perspective of the Nystr\"om approach \cite{alex2019operatortheoretic}, where it was shown that, under suitable ergodicity assumptions, the KAF estimator approximates the conditional expectation of observables of partially observed systems acted upon by the Koopman operator, thus yielding statistically optimal forecasts in the $L^2$ (root mean square error) sense. Here, we have shown that an analogous consistency property also holds if KAF is implemented using a one-parameter family of Markovian kernels in a limit of vanishing kernel bandwidth. The advantages of this smoothing approach over the Nystr\"om estimator are that it is positivity-preserving, and avoids the need for a kernel eigendecomposition. The latter carries the risk that the number of eigenfunctions needed to approximate the conditional expectation is larger than what can be feasibly computed, both in terms of computational cost and statistical robustness. This practical limitation of the Nystr\"om method is particularly prone to occur when the covariate space is high-dimensional, as we have demonstrated with numerical experiments involving a Hamiltonian system and the L96 system.  In problems were the regression function for the response (predictand) projects well onto the leading kernel eigenfunctions, the Nystr\"om method is still the method of choice, however.

Next, we studied the connection between KAF and the Mori-Zwanzig (MZ) framework for reduced dynamical modeling with memory. In particular, a major challenge in MZ approaches is to construct appropriate projection operators onto the covariate space (resolved dynamics), simplifying the structure of the memory kernel and rendering it amenable to approximation. We have argued that kernel methods provide a natural way of constructing improved projections by embedding the covariate data into a higher-dimensional space using delay-coordinate maps. This leads to a family of projections whose ranges are nested subspaces of the $L^2$ space associated with the invariant measure, increasing with the number of delays, and recovering the whole of  $L^2$ using finitely-many delays. Correspondingly, the MZ equation in this limit involves only a Markovian term, which was estimated here nonparametrically.

Our third topic of study was smoothing and predicting with noisy data. We proposed a scheme whereby delay-coordinate maps are used to obtain high-quality kernel eigenfunctions from noisy data, which are then used for denoising via subspace projection. Using again the L96 model as a testbed, we demonstrated that this nonparametric denoising scheme oftentimes outperforms classical state-estimation methods such as the ensemble Kalman filter and the 4D-VAR approach. Once denoised, the data can be used to train skillful forecast models via the Nystr\"om or smoothing formulations of KAF.      

Possible directions for future research include data-informed methods for kernel design that bias the leading eigenspaces of the corresponding integral operators such that they capture the response variable with minimal loss, thus improving the performance of Nystr\"om-based forecasting in high-dimensional covariate spaces. In addition, the efficacy of delay-coordinate maps in improving prediction skill of non-Markovian time series motivates the development of non-parametric kernel-based methodologies to estimate individual terms in the MZ equation.

\section*{Acknowledgments}

The research of JH was partially supported under the NSF grant DMS-1854299. DG acknowledges support from NSF grant 1842538, NSF grant DMS-1854383, and ONR grant N00014-16-1-2649. The authors thank Jonathan Poterjoy for sharing the EnKF and 4D-Var codes which he is developing as part of his NSF CAREER award.

\bibliographystyle{plain}

%\bibliography{arxiv_kme}

\end{document}